\documentclass[11pt,english]{article}
\usepackage[top=1in, bottom=1in, left=1in, right=1in]{geometry}
\usepackage{psfig,epsfig,fullpage,squeeze,wrapfig,amsmath,amssymb,float,titlesec,multirow,enumitem,hhline,makecell}
\usepackage[table]{xcolor}
\usepackage[T1]{fontenc}
\usepackage[colorlinks=true, linkcolor=blue]{hyperref}
\usepackage{bookmark}
\usepackage{titling,color,booktabs,caption,subcaption,tcolorbox, amsthm}
\usepackage{setspace}
\usepackage{rotating}

\usepackage[
  style=authoryear,
  maxcitenames=2,   
  mincitenames=1,
  maxbibnames=99,   
  minbibnames=99
]{biblatex}
\newtheorem{theorem}{Theorem}[section]
\newtheorem{lemma}{Lemma}

\ls{1.2}

\usepackage{algorithm}
\usepackage{algorithmic}

\newcommand{\beqa}{\begin{eqnarray}}
\newcommand{\eeqa}{\end{eqnarray}}
\newcommand{\beq}{\begin{equation}}
\newcommand{\eeq}{\end{equation}}
\newcommand{\ben}{\begin{enumerate}}
\newcommand{\een}{\end{enumerate}}
\newcommand{\bit}{\begin{itemize}}
\newcommand{\eit}{\end{itemize}}
\newcommand{\bi}{\begin{itemize} \item}
\newcommand{\ei}{\end{itemize}}

\newcommand{\QED}{\hfill $\Box$ \hfill}
\newtheorem{Definition}{Definition}[section]

\newtheorem{Theorem}{Theorem}[section]
\newtheorem{Proposition}[Theorem]{Proposition}
\newtheorem{Corollary}[Theorem]{Corollary}
\newtheorem{Example}{Example}[section]
\newtheorem{Lemma}[Theorem]{Lemma}

\newtheorem{Conjecture}{Conjecture}[section]

\newcommand{\begindef}{\begin{Definition} \rm}
\newcommand{\beginexa}{\begin{Example} \rm}
\newcommand{\beginthe}{\begin{Theorem} \rm}
\newcommand{\beginpro}{\begin{Proposition} \rm}
\newcommand{\beginlem}{\begin{Lemma} \rm}
\newcommand{\begincon}{\begin{Conjecture} \rm}
\newcommand{\begincor}{\begin{Corollary} \rm}

\newcommand{\eat}[1]{}

\def\papernumber #1 raised #2 {
\vspace{-#2}
\vbox to 0pt{\hfill\framebox{\bf Paper Number #1}}
\vspace{#2}
}

\renewcommand{\hat}{\widehat}

\begin{document}

\title{\Large {Networked Markets, Fragmented Data: Adaptive Graph Learning for Customer Risk Analytics and Policy Design}}
\vspace{-25ex}
\author{Lecheng Zheng, Jian Ni, Chris Zobel, John R Birge\thanks{Lecheng Zheng: Postdoc, Pamplin College of Business, Virginia Tech,  Email: lecheng@vt.edu; Jian Ni: Professor, Pamplin College of Business, Virginia Tech,  Email: jiann@vt.edu; Chris Zobel: Professor, Pamplin College of Business, Virginia Tech, Email: czobel@vt.edu; John R Birge: Professor, Booth School of Business, University of Chicago, Email: John.Birge@chicagobooth.edu}
}

\date{\vspace{-5ex}}

\maketitle

\fontsize{11pt}{11pt}\selectfont

\abstract{%

\noindent Financial institutions face escalating challenges in identifying high-risk customer behaviors within massive transaction networks, where fraudulent activities exploit market fragmentation and institutional boundaries. We address three fundamental problems in customer risk analytics: data silos preventing holistic relationship assessment, extreme behavioral class imbalance (<1\% suspicious transactions), and suboptimal customer intervention strategies that fail to balance compliance costs with relationship value. We develop an integrated customer intelligence framework combining federated learning, relational network analysis, and adaptive targeting policies. Our federated graph neural network enables collaborative behavior modeling across competing institutions without compromising proprietary customer data, using privacy-preserving embeddings to capture cross-market relational patterns. Focal loss optimization addresses class imbalance by amplifying learning from rare high-risk behavioral segments. We introduce cross-bank Personalized PageRank (PPR) to identify coordinated behavioral clusters—revealing fan-out, loop, and gather-scatter relationship structures—providing interpretable customer network segmentation for risk managers. A hierarchical reinforcement learning mechanism optimizes dynamic intervention targeting, calibrating escalation policies to maximize prevention value while minimizing customer friction and operational costs. Analyzing 1.4 million customer transactions across seven markets, our approach reduces false positive and false negative rates to 6.02\% and 5.53\%, substantially outperforming single-institution models. The framework prevents 79.25\% of potential losses versus 49.41\% under fixed-rule policies, with optimal market-specific targeting thresholds (0.316-0.523) reflecting heterogeneous customer base characteristics. Cross-market intelligence sharing increases precision in identifying high-risk customer networks by 60\%. These findings demonstrate that federated customer analytics materially improve both risk management effectiveness and customer relationship outcomes in networked competitive markets.
}

\vspace{1em} 
\textbf{Keywords:} Customer Analytics, Federated Learning, Graph Neural Networks, Financial Networks, Dynamic Targeting  

\newpage

\section{Introduction}
Financial institutions operate in an increasingly complex environment where identifying high-risk customer behaviors has become both a regulatory imperative and a strategic business challenge. As AI technologies reshape financial services and redefine business decision-making processes~(\cite{agrawal2019artificial}), financial institutions must balance predictive accuracy with operational constraints, regulatory compliance, and customer relationship value. Between 2020 and 2024, U.S. banks filed more than 137,000 suspicious activity reports linked to illicit networks~\parencite{website:FinCEN}, while recent investigation revealed that criminal organizations moved over \$312 billion through legitimate financial channels ~\parencite{website:WSJ}. Globally, the United Nations estimates that 2–5\% of world GDP (\$800 billion–\$2 trillion) is laundered annually, yet less than 1\% is ever recovered~\parencite{website:UN}. Beyond regulatory penalties, ineffective risk management imposes substantial operational costs through false alerts requiring manual investigation, potential reputation damage from compliance failures, and lost legitimate customer relationships due to over-cautious interventions. Understanding how well-intentioned financial regulations can produce unintended consequences is crucial for designing effective AML policies~\parencite{bao2017ni}.

Traditional customer risk management systems rely predominantly on rule-based heuristics that flag transactions exceeding predetermined thresholds or matching known suspicious patterns. While straightforward to implement, these systems generate overwhelming false positive rates—often exceeding 95\%—creating a critical resource allocation problem: compliance teams must investigate alerts that consume investigator time without commensurate risk reduction, while legitimate customers experience unnecessary friction that erodes relationship value. More critically, rigid rule-based approaches fail to adapt to evolving customer behavioral patterns and cannot detect sophisticated schemes that deliberately fragment activities across multiple institutions, jurisdictions, and time periods. As financial transactions become increasingly digital, interconnected, and global, the limitations of static rule-based customer monitoring have become economically costly and strategically problematic. 

Recent advances in machine learning offer promising alternatives for customer risk assessment. Data-driven approaches can process massive, complex behavioral datasets with efficiency and adaptability that traditional rules cannot match. By automatically identifying latent behavioral patterns and relationship anomalies, these methods can improve both detection accuracy and operational efficiency. However, despite their analytical promise, applying machine learning to customer risk management in financial markets faces four fundamental business and technical challenges that have limited widespread adoption. 

The first challenge is the fragmented customer intelligence across competing institutions. Financial institutions operate as competing entities, constrained by privacy regulations, competitive concerns, and customer confidentiality requirements. Customer transaction data are distributed across different entities and platforms, and institutions cannot easily share customer information due to strict regulatory constraints and competitive dynamics~\parencite{amoako2025exploring}. This fragmented view significantly limits the ability to capture illicit behavioral patterns that span multiple institutions—a deliberate strategy employed by sophisticated actors who exploit the lack of cross-institutional coordination. Models trained on isolated institutional data are prone to high false negatives, where suspicious activities appear benign within a single institution's customer base but reveal clear patterns when viewed across the broader market ecosystem. For instance, a transaction series may appear routine within one bank's customer portfolio but reveal layering or structuring patterns only when combined with behavioral data from other institutions. Addressing this challenge requires developing privacy-preserving collaborative frameworks that enable institutions to benefit from shared market intelligence without exposing proprietary customer information or violating competitive boundaries. 

The second challenge is lack of interpretability in customer risk models. Most advanced predictive models, particularly deep learning architectures, operate as "black boxes," making it difficult for relationship managers, compliance officers, and regulators to understand why particular customers or transactions are flagged. In high-stakes financial environments, model transparency is not merely desirable but mandatory for ensuring accountability, supporting customer relationship decisions, and satisfying regulatory auditing standards. When a model flags a corporate customer for suspicious wire transfers, investigators must trace the reasoning—such as anomalous transaction frequency, counterpart network characteristics, or deviations from historical behavioral patterns—to justify the alert and determine appropriate relationship management actions. Without interpretability, analysts face opaque predictions they cannot validate or contest, leading to mistrust and resistance toward model adoption. Moreover, regulators require financial institutions to explain the logic behind automated decisions, especially when these affect customer rights, relationship continuity, or trigger legal actions. The lack of interpretability thus impedes both internal investigation workflows and creates significant compliance and governance risks, ultimately limiting the business value of sophisticated analytics investments. 

Thirdly, detection is merely the first step in an effective customer risk management pipeline, as institutions must decide whether to monitor, escalate, or intervene in customer relationships under operational constraints and time pressures. Current systems typically employ rigid rule-based escalation procedures that fail to align with the probabilistic outputs of predictive models or optimize across multiple business objectives. A flagged transaction involving cross-border remittances might warrant different interventions depending on the customer's relationship value, transaction context, historical behavior patterns, and broader network risk exposure. Without adaptive decision-support mechanisms, relationship managers must manually reconcile probabilistic risk signals with static compliance rules, leading to inconsistent interventions, relationship friction, and suboptimal resource allocation. False positives—incorrectly flagging legitimate customers—can lead to unnecessary account restrictions, customer dissatisfaction, and relationship attrition, while false negatives—failing to identify genuine risks—result in financial losses and severe regulatory penalties. These challenges are compounded by institutional fragmentation, where each institution observes only a local subset of the broader customer relationship network. There is an urgent need for adaptive targeting frameworks that coordinate detection with intervention policies across institutional boundaries while optimizing for relationship value, operational costs, compliance requirements, and customer experience.

The fourth challenge is the extreme behavioral class imbalance. In real-world financial systems, illicit transactions constitute a minute fraction of total transaction volume—often less than 1\% as documented in Table~\ref{tab:country_graph_stats}—while the overwhelming majority corresponds to legitimate customer activities. This extreme imbalance fundamentally challenges model development and deployment. Standard machine learning algorithms become biased toward the majority (legitimate) class, achieving deceptively high overall accuracy while exhibiting poor sensitivity to rare but critical suspicious behaviors. Models may fail to detect subtle anomalies hidden within massive streams of normal customer transactions. Moreover, suspicious behaviors are highly heterogeneous and evolve over time, meaning the limited labeled positive cases cannot adequately capture the diversity of real-world schemes such as layering, structuring, or trade-based techniques. The scarcity of verified suspicious cases also reflects the costly and time-intensive nature of financial investigations, where labeling requires expert validation and cross-institutional cooperation. This imbalance complicates not only model training and evaluation but also operational deployment, as systems learn to overlook minority patterns that deviate from dominant behavioral trends.

To address these challenges, this paper proposes a privacy-preserving adaptive framework that connects three critical dimensions of modern financial security: illicit activities detection module, interpretable group identification, and hierarchical decision-making. Collectively, these innovations enable financial institutions to move from reactive detection toward proactive, coordinated, and explainable anti-money-laundering operations. Specifically, we first introduce a graph-based federated learning framework designed to detect money-laundering patterns while preserving data privacy. In this approach, each financial institution analyzes its own transaction network locally but shares abstracted, non-sensitive insights via \textit{virtual supernode} to corporately learn broader patterns of suspicious laundering behavior. This enables institutions to benefit from shared intelligence without exposing proprietary or customer data. The detection module also integrates multiple analytical objectives, such as identifying hidden transaction clusters and predicting anomalous relationships, to better capture the complexity of real-world laundering schemes. Instead of using the cross-entropy loss, we apply focal loss to mitigate the data imbalanced issue by dynamically adjusting the relative importance of the minority (suspicious) versus majority (normal) classes. The results show that the Cross-Entropy Loss tends to favor the majority class (normal transactions), leading to lower false alarm rates (Type I errors) but a higher tendency to overlook true laundering cases (Type II errors). In contrast, the Focal Loss introduces a more balanced learning process by assigning greater weight to these hard-to-detect, high-risk transactions, thus achieving much lower Type II errors and slightly higher Type I errors. We then evaluate the individual contribution of each component (e.g., virtual super-node, graph cut loss) through an ablation study, where components are systematically removed and replaced. The findings demonstrate that every component contributes meaningfully to the overall detection performance.

This paper reveals one fundamental challenge for cross-border financial crime detection: \textit{subpopulation shift} induced by data fragmentation resulting in performance degradation. \textit{subpopulation shift} is a form of distribution shift where the overall population remains fixed, but the characteristics or proportions of underlying subgroups differ~\parencite{koh2021wilds}. When transaction data are partitioned by country to reflect regulatory and privacy constraints, substantial distributional discrepancies emerge across markets, particularly for critical features such as transaction amount, which spans several orders of magnitude. Our experiments demonstrate that commonly used feature normalization strategies can inadvertently exacerbate this issue. Specifically, normalizing features independently within each country amplifies subpopulation shift, as country-level statistics differ markedly from global ones. Controlled experiments comparing country-level and global-level normalization show that this induced shift significantly degrades the performance of most machine learning models, leading to higher false negatives and lower precision-recall performance. Most approaches benefit from eliminating subpopulation shift via globally consistent normalization. These findings highlight subpopulation shift as a critical and previously underappreciated obstacle in AML modeling and motivate the need for learning framework that are robust to fragmented, privacy-constrained and heterogeneous financial data environments.

This paper also highlight several insights central to the effectiveness of detection module. First, federated learning generates substantial value from coordination without data sharing, a property particularly valuable when regulatory constraints prohibit direct information exchange. By aggregating model parameters rather than raw transactions, institutions achieve detection gains equivalent to pooled training while maintaining competitive independence. This reduces both Type I and Type II errors, demonstrating that shared model parameters (not shared data) deliver meaningful gains in detection performance. Second, models that incorporate structural information in graph data, such as federated graph neural network, consistently outperform feature-only deep learning methods, such as three-layer MLP. By leveraging relational signals that describe how accounts interact, the graph-based model better distinguishes legitimate from suspicious behaviors. 
Third, collaborative learning efficiency is critical: varying communication frequency in federated training demonstrates a non-monotonic trade-off between detection accuracy and synchronization cost. Moderate communication intervals improve AUPRC and reduce missed detections while lowering coordination overhead, whereas overly frequent or infrequent aggregation can degrade performance. 
Fourth, data fragmentation and subpopulation shift negatively impact detection quality: normalizing features independently within each country or training models solely on country-level subsets reduces performance compared with using aggregated or globally normalized data. To address this, we design a fixed-value min–max normalization strategy to mitigate subpopulation shift without sharing sensitive raw data across institutions. Overall, these findings underscore the importance of collaborative, privacy-preserving learning frameworks that can reconcile performance, structural modeling, and data heterogeneity in cross-country anti–money laundering detection.

Upon these detection results by the detection module, we further employ a network propagation mechanism via cross-bank Personalized PageRank (PPR) that starts from one suspicious account and expands the identification of suspicious neighbors entities, mimicking the process of money flows. This technique reveals sub-group structures within suspicious transaction networks, essentially identifying clusters of accounts or entities that act in coordination. By uncovering these latent community patterns, the module provides analysts and regulators with clearer explanations of what type of money-laundering behaviors are flagged, helping bridge the gap between algorithmic output and human reasoning. 
Together, these mechanisms turn abstract detection scores into interpretable group-level evidence, empowering financial analysts to act with greater confidence and precision. By examining the relational connections between accounts, this method uncovers hidden subgroups of potentially collusive actors and strengthens the overall detection accuracy. It ensures that the system not only flags individual transactions but also reveals the structure of laundering networks, which is a key step for investigative and regulatory follow-up. 

We further provide the theoretical insights of using PPR for grouping money laundering behaviors. Specifically, in Theorem~\ref{thm:detectability}, we come up with a detectability condition for identifying a money-laundering group using PPR. Intuitively, PPR captures money-laundering behavior because it propagates the suspiciousness signal from a known malicious seed node through the transaction network. Nodes densely connected to the seed, likely participating in group laundering, accumulate higher PPR scores, while weakly connected nodes receive lower scores. The theorem formalizes this by showing that a well-separated group will stand out in the PPR ranking. Based on this, we modify the transition matrix in PPR by assigning higher weights to edges more likely associated with laundering transactions, which further amplifies scores for nodes strongly connected via suspicious interactions. Specifically, weighted edges increase the transition probability along suspicious links, causing PPR to concentrate more on nodes involved in laundering.

This paper demonstrates the strategic value of the proposed Cross-bank Personalized PageRank (PPR) framework. By enabling institutions to integrate Cross-bank intelligence into the ranking process, Cross-bank PPR consistently identifies substantially more true malicious accounts while flagging fewer accounts overall, compared with conventional PPR. This outcome reflects a more conservative yet more precise detection strategy, reducing unnecessary investigations while elevating the identification of genuinely high-risk entities. This paper also reveals that the incorporation of the cross-bank suspiciousness signal markedly improves precision and reduces false positives by itself, whereas the edge-level maliciousness score degrades performance, when used alone to reweight transitions. We attribute it to the reliable prediction by the detection module. However, when the edge score is fused with the cross-bank signal inside Cross-bank PPR, its weaknesses are corrected and it meaningfully enhances detection. Together, these components significantly reduce false positives, increase the concentration of malicious accounts within discovered clusters, and improve detection consistency across countries. These findings highlight that cross-institutional signal sharing, when designed with strong privacy safeguards, can meaningfully strengthen financial crime detection beyond what any single institution can achieve in isolation.

Our analysis of interpretable group identification using cross-bank Personalized PageRank reveals that money-laundering behaviors often exhibit diverse structural patterns, including fan-out, loop, gather–scatter, and hybrid mechanisms that integrate multiple strategies. The proposed cross-bank PPR method effectively uncovers the pivotal accounts orchestrating these operations, even when illicit activities are distributed across institutions and national boundaries. In particular, our approach successfully reconstructs the complete illicit group behaviors for many laundering schemes. As laundering strategies become increasingly sophisticated by mixing multiple mechanisms and fragmenting transaction flows across borders, our method continues to accurately flag primary dispersal accounts and expose partial cross-border linkages, offering critical insights into the underlying operational structure. These findings underscore that analyzing relational patterns across merged clusters not only captures the organizational essence of money-laundering schemes but also enhances interpretability, providing actionable signals for risk assessment and regulatory intervention.


Finally, we propose a hierarchical reinforcement learning strategy to move from detection to action. Rather than treating AML as a static classification problem, this approach dynamically learns how to respond to suspicious activity. It considers factors such as network risk concentration, and operational costs to determine optimal intervention strategies: when to monitor, and when to escalate a freeze intervention. By balancing accuracy with cost and compliance priorities, this decision-making layer bridges the gap between automated detection and effective financial risk management. Our empirical evaluation shows that this decision-making module delivers substantial economic value. While a fixed-threshold policy prevents 50\% of illicit losses and exhibits extremely high Type II error rates, the hierarchical module prevents around 80\% of potential financial loss and dramatically reduces missed detections. Moreover, the learned thresholds vary across institutions in different countries, highlighting that optimal intervention policies are inherently context-dependent rather than universal. These findings underscore that effective AML risk management requires not just accurate prediction, but a policy that intelligently maps risk signals to timely and context-sensitive actions.

The remainder of this paper is organized as follows. We review the related works in Section 2 and Section 3 presents the proposed graph-based federated detection module, which enables privacy-preserving collaboration across institutions. Section 4 introduces the personalized PageRank and label propagation mechanisms used to uncover subgroup structures and refine detection accuracy. Section 5 details the hierarchical reinforcement learning framework that supports adaptive and coordinated decision-making. Section 6 presents the empirical results. Finally, Section 7 concludes the paper and discusses future research directions.

\section{Related Literature}
Illicit financial activity detection has been an active area of research for decades, encompassing a wide range of applications such as financial fraud detection~\parencite{fan2025unearthing, yiu2014deterrence, wu2024fair}, anti–money laundering detection~\parencite{deprez2025network}, e-commerce fraud detection~\parencite{nanduri2020microsoft, luca2016fake}, social security fraud detection~\parencite{van2017gotcha}, securities fraud detection in emerging markets~\parencite{guo2022hierarchical}, etc. With the growing scale and complexity of financial transactions, machine learning has become a dominant approach for detecting illicit activities, offering improved accuracy and adaptability compared to traditional rule-based systems~\parencite{august2025impact, fan2025unearthing, cecchini2010detecting, xiao2023black}. This study focuses specifically on the challenges of anti–money laundering detection.

The traditional ML-based AML systems are poorly equipped to handle this environment~\parencite{pettersson2022combating, alotibi2022money, jensen2023fighting, lokanan2024predicting}. Eric and Jannis \parencite*{pettersson2022combating} evaluate four classic supervised learning algorithms (e.g., logistic regression, random forest, support vector machines, and decision trees) on the AML task and find that these models are too slow and require further optimization for the anti–money laundering context. More advanced supervised machine learning models have been proposed, trained on historical transaction data to detect anomalies and laundering patterns, offering more accurate and adaptive detection than traditional rule-based systems~\parencite{golightly2023securing}. However, these methods rely on static machine learning and deep learning models, which fail to integrate information across institutions and struggle to adapt to temporally complex laundering strategies. Furthermore, studies show that some laundering schemes are executed rapidly over a handful of transactions, while others evolve over months and mix tactics such as layering and cross-border transfers~\parencite{amoako2025exploring}. Anti–money laundering detection models tuned to one temporal horizon often miss others, creating a persistent tension between premature intervention and delayed action. Wang et al. \parencite*{wang2024temporal} demonstrate that traditional anti–money laundering systems struggle to detect complex schemes that leverage cross-jurisdictional gaps and the temporal dispersion of transactions, thus failing to capture complex evolving patterns in transaction networks. 

Recently, Kurshan and Shen~\parencite*{kurshan2020graph} and Lokanan~\parencite*{lokanan2023financial} show the prominent performance of graph-based approaches on general financial fraud. For instance, Kurshan and Shen~\parencite*{kurshan2020graph} provide a high-level discussion of the challenges faced by machine learning methods in combating financial crime, emphasizing issues such as scalability, data quality, and operational constraints. Although their survey does not center on AML, they underscore the strong potential of network-based methods for modeling the large-scale transactional data handled by financial institutions. Complementary to these methodological surveys, Gerbrands et al.~\parencite*{gerbrands2022effect} examine AML from a policy and network-structure perspective, demonstrating how regulatory interventions reshape transactional networks and influence illicit financial behaviors. While this line of work highlights the importance of network analysis for understanding money-laundering dynamics, it does not develop learning-based detection models. More closely aligned with graph neural network approaches for AML, Kute et al.~\parencite*{kute2022explainable, kute2024explainable} demonstrate how deep graph learning models can be augmented with explanation mechanisms to clarify which transaction patterns and behavioral attributes drive model outputs.

Model interpretability is critical in financial applications, as compliance officers and investigators must understand, trust, and act upon the model’s recommendations. Contemporary work on interpretable AML detection spans methodological advances in feature-based models, post-hoc explanation techniques, and network-aware methods and much of this literature appears in business-oriented outlets. Early applied studies established the value of supervised learning for prioritizing alerts while emphasizing model transparency~\parencite{chen2018machine, han2020artificial}. Building on these works, Jullum et al. \parencite*{jullum2020detecting} develop and validate machine-learning scoring systems for transaction prioritization, showing how transparent feature design and risk-scoring improve investigatory efficiency. Konstantinidis and Gegov \parencite*{konstantinidis2024deep} evaluate explainable AI tools (e.g., SHAP and LIME) as operational supplements to black-box models, demonstrating that feature-attribution explanations improve auditability and the quality of Suspicious Activity Reports. Oloyede \parencite*{oloyede2025evaluating} explores how interpretability (e.g., transparent ML models) affects regulatory compliance, trust, and operational effectiveness in financial institutions.  In parallel, Gu, Yang, and Liu \parencite*{gu2025optimization} incorporates causal reasoning to improve both the transparency and the decision consistency of AML systems. Recent reviews~\parencite{mazumder2025explainable} synthesize these strands and call explicitly for XAI and socio-technical integration so that automated detection maps cleanly onto operational workflows and regulatory expectations. Collectively, these works show that effective AML solutions must go beyond predictive accuracy to provide clear, auditable rationales for their decisions, thereby enhancing investigator trust, supporting compliance verification, and ensuring that automated detection remains aligned with financial governance standards.

In the past decades, a growing body of research examines how anti–money laundering decisions and policies are shaped by regulatory frameworks, organizational practices, and behavioral responses to financial crime risks. Early work \parencite{ross2007money} highlights the central tension in AML supervision between rule-based requirements and risk-based decision-making, showing how compliance officers must interpret ambiguous regulations under conditions of limited information. Unger and Van Waarden \parencite*{unger2009dodge} further compare rule-based and risk-based AML strategies, arguing that excessive data collection without effective prioritization can overwhelm institutions and undermine detection quality. Complementing these policy-oriented perspectives, McCarthy, van Santen, and Fiedler \parencite*{mccarthy2015modeling} develop a microtheoretical model of money launderer behavior, illustrating how adversarial adaptation should inform AML policy design. Similarly, research on strategic learning behavior in markets demonstrates that sophisticated actors rapidly adapt their strategies in response to observed outcomes and platform policies~(\cite{xu2022ni}), underscoring the importance of adaptive detection systems that evolve alongside increasingly sophisticated laundering tactics. Helgesson and Mörth \parencite*{helgesson2016involuntary} shift attention to the role of professional intermediaries, demonstrating that lawyers become de facto policy actors when implementing AML and counter-terrorism financing mandates, often shaping regulatory outcomes through everyday compliance decisions. More recently, methodological advances have emerged, such as the reinforcement learning framework proposed by \parencite{rao2025reinforcement}, which models cross-border transaction anomalies from a behavioral-economics perspective and offers a data-driven approach to dynamic AML decision-making. Together, these studies reveal an evolving landscape in which AML effectiveness depends not only on formal rules and technologies but also on how institutions interpret risks, allocate attention, and adapt to strategic adversaries.

Despite significant progress, these AML detection methods suffer from several critical limitations that prevent them from keeping pace with increasingly complex laundering behaviors. Traditional supervised machine learning and deep learning models~\parencite{pettersson2022combating, alotibi2022money, jensen2023fighting, lokanan2024predicting, chen2018machine}, while effective in narrow settings, remain largely static, siloed, and institution-specific. These systems struggle to generalize across heterogeneous financial institutions, cannot integrate cross-bank intelligence, and often fail under data imbalance or rapidly evolving laundering tactics. Recent interpretable or explainable AI approaches~\parencite{chen2018machine, han2020artificial, kute2024explainable, kute2022explainable} focus primarily on feature-level transparency and do not reveal coordinated group behaviors or network structures underlying illicit flows. At the policy and decision-making level, current systems~\parencite{ross2007money, unger2009dodge, helgesson2016involuntary, mccarthy2015modeling} are still reactive, offering little guidance on when and how local institutions should intervene in unfolding laundering schemes under a global coordination system. As a result, AML operations remain fragmented: models detect anomalies, investigators manually infer group structures, and compliance teams determine actions without principled support. The privacy-preserving adaptive framework proposed in this paper directly addresses these gaps. By enabling graph-based federated learning, our method preserves institutional privacy while facilitating cross-bank sharing. The integration of focal-loss–based detection, cross-bank PPR based group identification, and hierarchical reinforcement learning creates a unified pipeline that not only identifies suspicious entities but also uncovers their coordinated structures and recommends cost-sensitive, context-aware interventions. In doing so, the proposed system advances AML practice from isolated, static prediction tools toward proactive, explainable, and strategically optimized financial crime prevention.

\section{Detection module}
\label{sec:Method}

\subsection{Graph Construction}
Financial fraud, particularly money laundering~\parencite{chen2018machine}, poses a major risk to the stability and integrity of financial systems. A fundamental challenge in detecting money laundering is that transactional data are fragmented across banks due to privacy concerns, preventing the construction of a global view of illicit activities~\parencite{le2018preventing, van2024privacy}. This bank isolation limits existing detection systems, leading to both false positives and false negatives, as suspicious behaviors often only emerge when transactions across multiple institutions are jointly analyzed~\parencite{wang2024temporal, amoako2025exploring}. However, directly aggregating transaction data across banks is infeasible due to privacy and regulatory constraints~\parencite{thommandru2023exploring}. 

Traditional machine learning approaches typically treat each transaction as an independent observation and largely ignore the relationships among consecutive money-laundering transactions made by the same group of accounts~\parencite{wang2024temporal, amoako2025exploring, pettersson2022combating, alotibi2022money, jensen2023fighting, lokanan2024predicting}. Such independence assumptions limit their ability to recognize coordinated or evolving laundering schemes, where illicit funds are often transferred through a sequence of interconnected entities to obscure their origin and ownership. Inspires by this observation, we represent the transaction records as a set of graphs in which nodes represent accounts and edges denote fund transfers between them. The conversion to graph-structured data allows us to capture the interdependencies among transactions both direct and indirect within and across laundering groups.

Consider a transaction network\footnote{We use the terms \textit{graph} and \textit{network} interchangeably.} $\mathcal{G}=\{\mathcal{G}_1, \mathcal{G}_2, \ldots, \mathcal{G}_m\}$ consisting of $m$ subgraphs, where each subgraph $\mathcal{G}_i=\{\mathcal{V}_i, \mathcal{E}_i, X_i^V, X_i^E, Y_i\}$ represents the financial transaction system within a specific country. The set of nodes $\mathcal{V}_i (|\mathcal{V}_j|=n_i)$ consists of $n_i$ entities (such as individual or corporate accounts) operating in country $i$, while the set of edges $\mathcal{E}_i \subseteq \mathcal{V}_i \times \mathcal{V}_j $ captures the transaction links among these entities. When an edge connects $\mathcal{V}_i$ and $\mathcal{V}_j$ with $i \neq j$, it represents a cross-border transaction, reflecting international fund transfers between entities in different jurisdictions. $Y_i$ is the label vector, where $Y_{i,j}=1$ indicates that the $j$-th transaction is a money-laudering behavior and $Y_{i,j}=0$, otherwise. Here, we convert the set of edges $\mathcal{E}_i$ into an adjacency matrix $A_i\in \mathbb{R}^{n_i \times n_i}$, where $A_i(j,k)=1$ indicates an existing transaction link and  $A_i(j,k)=0$, otherwise. Each node is associated with an attribute matrix $X_i^V\in \mathbb{R}^{n_i \times d_i}$ that records its observable financial or behavioral characteristics, such as transaction frequency, while each edge is characterized by an attribute matrix $X_i^E$ that describes the properties of a transaction, including payment amount, method, and timing. Ultimately, our objective is to identify money-laundering activities, represented by malicious transaction links $\mathcal{E}_i$ among accounts, and to understand how these illicit behaviors manifest and propagate across different countries.

\subsection{Federated Graph-based Detection module}
\label{sec:detection_model}
Given the adjacency matrix $A_i$, we employ a graph neural network (GNN)–based encoder~\parencite{kipf2016semi, DBLP:conf/iclr/XuHLJ19} to map the pair $(A_i, X_i^V)$ into a set of latent representations at the $L$-layer:
\begin{equation}
    H_i^{L+1} = \sigma(A_i H_i^L W^L),
\end{equation}
where $H_i^0=X_i^V$ is the raw node attribute at the first layer, $W^L$ is the weight matrix at the $L$-th layer, and $\sigma(a)=\max(0,a)$ is the ReLU activation function. Intuitively, this graph encoder aggregates information from each entity’s local neighborhood, allowing the representation of an account to be informed not only by its own attributes but also by the behaviors and characteristics of the entities it transacts with. This relational learning process captures how risk exposure or suspicious activity propagates through the transaction network. In financial systems, money-laundering behaviors often span multiple accounts and evolve through multi-step transaction chains. To emulate such interlinked and layered behaviors, we concatenate the latent representations learned from all layers to form the final embedding $H_i = [H_i^0, H_i^1, \ldots, H_i^L]$. This multi-level aggregation allows the module to integrate both local and global structural information, thereby capturing higher-order dependencies that reflect complex transaction patterns, such as circular money flows, bursty transfers, or structured layering activities commonly used to disguise illicit origins. 
To transform node-level representations into edge-level embeddings for transaction classification, we design an edge construction function that aggregates the features of the source and destination nodes together with the edge-specific attributes before computing the prediction via a multi-layer perceptron (MLP). Specifically, given two latent representations $h_i \in \mathbb{R}^d$, $h_j \in \mathbb{R}^d$ for nodes $i$, $j$, and the edge $\mathcal{E}_{i,j}$ associated attribute vector $X^E_{ij} \in \mathbb{R}^k$, we construct the edge representation $z_{ij}$ as
\begin{equation}
    z_{ij} = \left[\, h_i \,\|\, (h_j - h_i) \,\|\, X^E_{ij} \,\right],
\end{equation}
where $\|$ denotes the vector concatenation operation. In the representation $z_{ij}$, $h_i$ encodes the intrinsic behavioral and structural characteristics of the source account; $(h_j - h_i)$ models the directional discrepancy between the destination and source accounts, thereby encoding the asymmetric flow of information or funds, which is crucial for detecting suspicious transfers; $X^E_{ij}$ preserve contextual information specific to the transaction, such as the amount or transaction type, providing direct evidence for abnormal activity. Overall, the resulting vector $z_{ij}$ integrates both node-level and edge-level information into a unified representation that characterizes each transaction from multiple perspectives. This representation is then passed through a multi-layer perceptron $f_{\theta}(\cdot)$ followed by a sigmoid activation to predict the probability of the transaction being illicit:
\begin{equation}
    \hat{y}_{ij} = \sigma ( f_\theta(z_{ij})),
\end{equation}
where $\sigma(\cdot)$ denotes the sigmoid function. The model can be trained end-to-end using the binary cross-entropy loss between the predicted probability $\hat{y}_{ij}$ and the ground-truth label $y_{ij}$:
\begin{equation}
    \mathcal{L}_{\text{BCE}} = - \frac{1}{|\mathcal{E}|} \sum_{(i,j) \in \mathcal{E}} \Big[\, y_{ij} \log(\hat{y}_{ij}) + (1 - y_{ij}) \log(1 - \hat{y}_{ij}) \,\Big].
\end{equation}
Through this formulation, the model learns to distinguish illicit transactions by jointly considering node characteristics, their directional interactions, and edge-level contextual evidence. 
However, in real-world financial networks, the distribution of transactions is highly skewed as the vast majority are legitimate, and only a minute fraction correspond to money-laundering activities. This extreme class imbalance poses a fundamental learning challenge because the model tends to be biased toward predicting normal behavior, as minimizing cross-entropy naturally favors the dominant class. Consequently, even a trivial classifier that labels every transaction as normal can achieve deceptively high accuracy but fail to identify the rare yet critical illicit links that drive regulatory risk. To mitigate this imbalance, we adopt the focal loss~\parencite{lin2017focal}, which dynamically reweights training samples to focus more on hard-to-classify, minority cases. The focal loss modifies the cross-entropy objective as follows:
\begin{align}
\label{focal_loss}
    \mathcal{L}_i^{FL} = - \sum_j \big[\alpha(1 - \hat{\mathbf{y}}_j)^{\gamma} y_j \log(\hat{\mathbf{y}}_j) + (1-\alpha)\hat{\mathbf{y}}_j^{\gamma} (1-y_j) \log(1-\hat{\mathbf{y}}_j)\big],
\end{align}
where $\alpha \in [0,1]$ is a class-balancing factor that adjusts the relative importance of the minority (suspicious) versus majority (normal) classes, and $\gamma > 0$ is a focusing parameter that controls the attention on misclassified or uncertain samples. Intuitively, $\alpha$ serves as an ``attention weight'' for class imbalance. A higher $\alpha$ increases the importance of correctly identifying suspicious transactions, ensuring that rare but high-impact money-laundering cases receive more focus, whereas a lower $\alpha$ reduces sensitivity to the minority class. The focusing parameter $\gamma$ emphasizes hard-to-classify samples. For a positive example ($y_j = 1$), when the model is confident (\(\hat{\mathbf{y}}_j \approx 1\)), the factor $(1-\hat{\mathbf{y}}_j)^\gamma$ downweights the loss contribution, reducing the influence of easy cases. Conversely, for uncertain or misclassified positive examples (\(\hat{\mathbf{y}}_j \ll 1\)), the loss remains high, drawing attention to these high-risk transactions. Similarly, for negative examples ($y_j = 0$), the factor $\hat{\mathbf{y}}_j^\gamma$ ensures the model focuses on ambiguous normal transactions. Collectively, $\alpha$ and $\gamma$ enable the model to prioritize learning from rare and high-risk transactions without being overwhelmed by the abundant routine activity.

From a managerial perspective, these hyperparameters encode distinct dimensions of institutional risk tolerance. The class-balancing factor $\alpha$ reflects the strategic trade-off between compliance and customer experience: higher $\alpha$ prioritizes catching illicit transactions even at the cost of more false positives, while lower $\alpha$ minimizes unnecessary account friction to preserve customer relationships. The focusing parameter $\gamma$ controls operational attention allocation: higher $\gamma$ concentrates learning on ambiguous, hard-to-classify transactions—precisely the borderline cases that consume investigator time—while lower $\gamma$ treats all misclassifications more uniformly. This parameterization allows institutions to calibrate detection systems to their specific regulatory environment and operational capacity.

As discussed earlier, one of the most critical challenges in cross-border money-laundering detection lies in the bank-isolation problem, where financial institutions are prohibited from directly sharing transactional data across borders due to strict regulatory and privacy constraints. This limitation hinders the ability to identify complex laundering schemes that intentionally distribute suspicious transfers across multiple banks or jurisdictions. To address this challenge while maintaining data confidentiality, we introduce the concept of a \textit{virtual super-node} $\mathcal{V}^S_{uv}$ ($u \in \mathcal{V}_i, v \in \mathcal{V}_j$), which abstractly represents the transactional relationship between two accounts belonging to different financial institutions. The super-node serves as a ``boundary connector'' between two subgraphs $\mathcal{G}_i$ and $\mathcal{G}_j$, enabling institutions to share structured insights rather than raw transaction details. Formally, we denote by $\mathcal{V}^S$ the set of virtual super-nodes, i.e., $\mathcal{V}^S_{uv} \in \mathcal{V}^S$. Each super-node summarizes localized transaction behavior in a privacy-preserving form that can be exchanged for collaborative analysis. 
Specifically, we define the representation of a super-node as: 
\begin{equation}
    h_{\mathcal{V}^S_{iu}} = \frac{1}{|N_{iu}|}\sum_{v \in N_{iu}} h_v, \quad
    h_{\mathcal{V}^S_{jv}} = \frac{1}{|N_{jv}|}\sum_{u \in N_{jv}} h_u,
\end{equation}
where $h_{\mathcal{V}^S_{iu}}$ and $h_{\mathcal{V}^S_{jv}}$ capture aggregated local neighbor information from their neighbor set $N_{iu}$ and $N_{jv}$ in subgraphs $\mathcal{G}_i$ and $\mathcal{G}_j$, respectively. Notice that we can extend the mean aggregation (first-order statistical moment) to the combination of first-order and high-order moment, such as the aggregation of mean and variance, while our empirical results in Section~\ref{supernode_high_order_results} only show the marginal effect of including high-order moment. These embeddings act as compact statistical summaries of local cross-border activities, allowing each institution to participate in a federated learning process without revealing any sensitive data. To ensure semantic consistency across institutions, we introduce a self-consistency loss:
\begin{equation}
    \mathcal{L}_i^{SC} = \sum_{\mathcal{V}^S_{uv} \in \mathcal{V}^S_i} \| h_{\mathcal{V}^S_{iu}} - h_{\mathcal{V}^S_{jv}} \|_2^2,
\end{equation}
where $\mathcal{V}^S_i$ is the subset of super-nodes associated with $\mathcal{G}_i$. This objective encourages alignment between the embeddings of shared super-nodes across banks, ensuring that the same cross-institutional transaction pattern is interpreted consistently. Conceptually, it enables a federated network of banks to ``agree'' on the structural signature of suspicious behavior, thereby strengthening the global defense against coordinated money-laundering schemes while respecting jurisdictional boundaries.

Beyond individual transaction links, money-laundering often manifests in higher-order structures, such as rings, chains, or star-shaped hubs that circulate illicit funds across multiple accounts. To explicitly encourage the model to capture such patterns, we employ a soft-membership differentiable graph cut loss~\parencite{nazi2019gap}, which promotes the identification of dense, irregular clusters in a fully differentiable manner:
\begin{equation}
\mathcal{L}_i^{GCL} = -\frac{\mathrm{trace}(\mathcal{M}_i^\top A_i \mathcal{M}_i)}{\mathrm{trace}(\mathcal{M}_i^\top D_i \mathcal{M}_i)} + \beta \|\mathcal{M}_i^\top \mathcal{M}_i - I\|_F^2,
\end{equation}
where $\mathcal{M}_i = f_\gamma(A_i, H_i)$ represents the soft community assignment of nodes in subgraph $\mathcal{G}_i$, $D_i$ is the diagonal degree matrix of $A_i$, $I$ is the identity matrix, and $\beta$ is a regularization hyperparameter. The first term encourages nodes with strong transactional connections to be grouped together in the same cluster, effectively capturing tightly knit substructures that may indicate coordinated money-laundering activity. The second term regularizes the assignments to promote orthogonality among different clusters, ensuring that each node’s membership is distributed consistently and that communities remain distinct. Intuitively, this loss guides the embeddings to highlight anomalous clusters, such as small rings or star-shaped hubs, that deviate from normal transaction patterns, surfacing laundering groups that may evade detection through standard anomaly scoring alone. Finally, each bank optimizes a multi-objective loss: 
\begin{align}
\label{overall_loss}
    \mathcal{L}_i &= \mathcal{L}_i^{FL} + \lambda_1 \mathcal{L}_i^{GCL} + \lambda_2 \mathcal{L}_i^{SC},
\end{align}
where $\mathcal{L}_i^{FL}$ is the focal loss for edge prediction under weak supervision to distinguish suspicious transaction links, $\mathcal{L}_i^{GCL}$ is the soft-membership differentiable graph cut loss that encourages discovery of dense and anomalous clusters indicative of coordinated laundering schemes, and $\mathcal{L}_i^{SC}$ is the cross-bank consistency loss that aligns shared super-node representations. The hyperparameters $\lambda_1$ and $\lambda_2$ balance the relative importance of the community structure and cross-bank alignment objectives, allowing banks to jointly learn risk-sensitive embeddings while preserving privacy and regulatory compliance.


Building upon the federated learning paradigm, each bank optimizes its local objective to update its model parameters $\theta_i^{local}$ using only its internal transaction data. Critically, raw data never leaves the institution, preserving privacy and regulatory compliance. At predefined communication rounds, each bank transmits only its model updates (gradients or parameters) and the aggregated super-node embeddings to a central server, which acts as a coordination hub rather than a data repository. Following the widely used federated averaging protocol~\parencite{zhang2021survey, li2020review}, the server updates the global model parameters as
\begin{equation}
    \theta^{global} \leftarrow \frac{1}{m} \sum_{i=1}^{m} \theta_i^{local},
\end{equation}
where $m$ is the total number of participating banks. This weighted averaging effectively combines the knowledge learned from heterogeneous local datasets while mitigating biases that may arise from imbalanced or institution-specific transaction patterns. The updated global parameters $\theta^{global}$ are then redistributed to all participating banks, allowing each institution to benefit from the collective intelligence of the network without ever accessing another bank’s raw transactions. Intuitively, this collaborative process enables banks to detect complex, distributed money-laundering schemes that span multiple institutions or jurisdictions. By leveraging federated learning, the system simultaneously improves detection accuracy, respects data privacy, and aligns with regulatory constraints, creating a practical and scalable framework for cross-border financial crime prevention.

\section{Identifying Group Patterns with Cross-bank Personalized PageRank}
\subsection{Personalized PageRank}
While graph-based detection modules achieve strong predictive performance in identifying suspicious transactions, their complexity often limits interpretability, making it challenging for compliance officers to understand why certain accounts are flagged. In our proposed method, we employ a differentiable graph cut loss to capture higher-order structures, such as rings, chains, or star-shaped hubs that circulate illicit funds across multiple accounts. However, this approach alone may not capture group-level suspicious behaviors spanning multiple countries or accounts that often represent a mixture of incomplete laundering schemes rather than compact local structures. To enhance interpretability at the group-level, we incorporate Personalized PageRank (PPR)~\parencite{page1999pagerank}, a well-established algorithm originally developed and deployed by Google in the early 2000s to rank and organize web pages in its search engine. In the web context, PPR models how user interest in a particular page propagates through hyperlinks to identify other closely related pages. This intuitive notion of relevance diffusion motivates its application in financial networks. Specifically, starting from known high-risk accounts, PPR enables us to trace and quantify how suspicion propagates through transactional links to other closely connected accounts. By diffusing risk signals along multi-hop transaction paths, PPR reinforces model predictions while producing transparent and interpretable local clustering results that highlight intermediary accounts, relational dependencies, and latent laundering structures. This relational perspective allows compliance officers to better understand the contextual basis of flagged accounts, thereby increasing trust in model outputs and supporting more targeted and actionable investigations.


Formally, let the transactional network in a country be represented as a directed graph $\mathcal{G}_i$, with a weighted adjacency matrix $P$, where $P_{ij}$ reflects the normalized transaction intensity (e.g., frequency, volume, or dollar amount) from node $i$ to node $j$. Traditional rule-based systems often focus on local anomalies, such as unusually large transfers or sudden spikes in activity. In contrast, money-laundering schemes are inherently relational, spanning multiple transactions, intermediaries, and jurisdictions. PPR addresses this limitation by diffusing initial suspicion scores across the network topology, allowing latent patterns, collusive groups, and hidden intermediaries to surface naturally.
Formally, the PPR vector $\boldsymbol{r} \in \mathbb{R}^{|V|}$ represents a steady-state distribution of a random walk over the network. At each step, the walk either follows an outgoing transaction with probability $(1-\alpha_{PPR})$ or teleports back to a set of high-suspicion seed nodes with probability $\alpha_{PPR}$:
\begin{equation}
    \boldsymbol{r} = (1 - \alpha_{PPR})P \boldsymbol{r} + \alpha_{PPR} \boldsymbol{v},
\label{eq:ppr}
\end{equation}
where $\alpha_{PPR} \in (0,1)$ is a factor that balances local versus global influence, and $\boldsymbol{v}$ is the personalization vector that concentrates probability mass on a set of known high-risk accounts $Q \subset \mathcal{V}_i$. Specifically:
\begin{equation}
v_i = 
\begin{cases}
    \frac{\hat{\mathbf{y}}_i}{\sum_{j \in Q} \hat{\mathbf{y}}_j}, & \text{if } i \in Q, \\
    0, & \text{otherwise,}
\end{cases}
\end{equation}
where $P_i$ is the initial suspiciousness score of node $i$, derived from prior models, expert labeling, or regulatory flags.

\subsection{Cross-bank Personalized PageRank}
While PPR effectively amplifies and contextualizes suspiciousness signals within a single financial institution, it faces inherent limitations in the cross-border setting. Because of strict privacy and data-sharing regulations, transaction networks are fragmented across banks and jurisdictions, preventing the propagation of suspicion signals across institutional boundaries. Local biases in financial markets, documented in contexts such as crowdfunding~\parencite{ni2025xinyang}, suggest that institutions may systematically underweight cross-border signals, further motivating the need for mechanisms like Cross-bank PPR that explicitly integrate geographically distributed intelligence. As a result, money-laundering activities spanning multiple countries may appear incomplete or innocuous when analyzed in isolation as each country observes only a fragment of the global laundering pattern.  

The effectiveness of PPR-based detection depends fundamentally on the separability of money-laundering groups within the network topology—a property we formalize in Theorem \ref{thm:detectability}. To enhance detectability in fragmented settings, we introduce Cross-bank PPR, which incorporates cross-institutional intelligence while maintaining privacy constraints. Specifically, we modify the standard PPR formulation to incorporate the cross-bank suspicious signal from other countries into PPR design. 
Formally, by incorporating the cross-bank component $\boldsymbol{v}_{cross}$, the formulation of PPR can be updated as follows:
\begin{align}
    \boldsymbol{r} &= (1 - \alpha_{PPR})P \boldsymbol{r} + \alpha_{PPR} (\boldsymbol{v}_{local} + \label{PPR_weight} \boldsymbol{v}_{cross}), \\
    P &= D^{-1}(A + \hat{A})  
\label{eq:ppr}
\end{align}
where $\hat{A}_{ij}=\hat{y}_{ij}$ is the prediction of the edge $\mathcal{E}_{ij}$ and $\hat{A}_{ij}=0$ if the edge $\mathcal{E}_{ik}$ does not exist, $D$ is the degree matrix of $(A + \hat{A})$,  $\boldsymbol{v}_{local}$ encodes the bank’s internal suspicion sources, and $\boldsymbol{v}_{cross}$ represents external suspicious signals towards the cross-bank transaction shared across banks. Different from traditional PPR, we define $P$ the normalized weighted adjacency matrix that captures transaction intensity and encodes the suspicious score of the edge through $A$ and $\hat{A}$ respectively.
The teleportation factor $\alpha_{PPR}$ thus determines the balance between exploring local transactional neighborhoods and returning to known suspicious nodes. By incorporating the cross-bank component $\boldsymbol{v}_{cross}$, we enable a privacy-preserving exchange of aggregated suspicion information, allowing each bank to indirectly benefit from the intelligence gathered by others without revealing raw transactions. This cross-institutional diffusion enriches each bank’s PPR computation, helping surface hidden multi-bank laundering rings that would otherwise remain undetected due to regulatory isolation. In this way, the cross-bank PPR formulation combines fragmented data while maintaining compliance, substantially enhancing both the interpretability and completeness of risk detection.

In practice, the Cross-bank PPR vector can be computed iteratively as:
\begin{align}
    \boldsymbol{r}^{(t+1)} &= (1 - \alpha_{PPR})P \boldsymbol{r}^{(t)} + \alpha_{PPR} (\boldsymbol{v}_{local} + \boldsymbol{v}_{cross})  \\
    \boldsymbol{r}_{\text{n}} &= \frac{\boldsymbol{r}^{(t+1)}}{\text{max}(\boldsymbol{r}^{(t+1)})}, 
\end{align}
which converges to a unique solution under mild conditions. Since the largest score in $\boldsymbol{r}^{(t+1)}$ corresponds to the most suspicious seed node with the highest likelihood of being anomalous, we normalize $\boldsymbol{r}^{(t+1)}$ by dividing each node’s score by the maximum value in $\boldsymbol{r}^{(t+1)}$, thereby converting the raw scores into anomaly probabilities. Using a predefined threshold, nodes with sufficiently high probabilities are identified as suspicious. Intuitively, in the PPR, nodes that are both structurally proximate to and frequently transacting with high-suspicion accounts accumulate higher steady-state probabilities. The PPR spreads risk signals through the network: even accounts not directly involved with known suspicious actors may receive elevated scores if they participate in multi-step chains, loops, or intermediary hubs commonly used in layering and integration stages of money laundering. Notice that high-ranking nodes can be traced through their most probable propagation paths back to seed nodes, offering analysts narrative explanations such as “Account $v$ connects to seeded node $Q$ through a series of small but structured transfers.” These interpretable trails facilitate compliance documentation and regulatory auditing. 
Overall, PPR provides a mathematically grounded and operationally transparent mechanism to uncover money laundering rings by diffusing known suspicion through the transactional network. By integrating graph-theoretic reasoning with domain-aware weighting and temporal adaptation, it bridges the gap between static anomaly scoring and dynamic, relationship-centric risk discovery, thus improving both the precision and explainability of anti-money-laundering analytics.
 
To obtain a more comprehensive view of money-laundering activities across multiple countries, we merge clustering results from different national transaction networks into a unified framework. We start by identifying the cluster dictionary with the largest number of accounts as the initial reference. For each cluster in this reference dictionary, we iteratively examine clusters from other country-specific dictionaries to detect overlapping accounts. When overlaps are found, the clusters are merged to form a larger, consolidated cluster, ensuring that duplicated accounts are avoided. This process is repeated iteratively across all clusters, further merging any partially overlapping clusters to capture extended relational structures. By combining clusters in this way, the resulting merged dictionary integrates patterns spanning multiple financial institutions and geographies, allowing us to recover more complete money-laundering behaviors that may be fragmented across different national datasets. The merged clusters are then cross-validated against known laundering attempts to highlight clusters that contain multiple suspicious accounts, enhancing the interpretability and practical utility of the results.

\subsection{Theoretical Insights of Using PPR for Grouping Money Laundering Behaviors}
Having introduced Cross-bank PPR and demonstrated its empirical effectiveness, we now provide theoretical foundations for understanding when and why PPR-based methods can successfully identify money-laundering groups. We show that Personalized PageRank (PPR) initialized from a malicious seed node concentrates its probability mass inside the laundering group whenever the within-group connectivity is sufficiently stronger than the background connectivity, following the analysis by Avrachenkov and Andersen et al.~\parencite{avrachenkov2015pagerank, andersen2006local}. Let $\mathcal{G}=(\mathcal{V},\mathcal{E})$ denote a random graph  generated by a two-block Stochastic Block Model (SBM) with $n$ nodes (we omit the graph index in this theoretical analysis for notation simplicity) and a planted laundering group $S^\star\subset V$ of size $s$. Edges are sampled independently as
\begin{equation}
\Pr[(u,v)\in \mathcal{E}] =
    \begin{cases}
    p_{\mathrm{in}}, & u,v\in S^\star,\\[0.2em]
    p_{\mathrm{out}}, & \text{otherwise},
    \end{cases}
\end{equation}
Let $A$ be the adjacency matrix, $D$ the degree diagonal, and $P=D^{-1}A$ the random-walk transition operator. Given $\alpha\in(0,1)$ and a one-hot seed vector $\mathbf{e}_{s_0}$ (corresponding to a known malicious seed $s_0\in S^\star$), the Personalized PageRank vector is
\begin{equation}
    \mathbf{r} =(1-\alpha)\,(I-\alpha P)^{-1}\mathbf{e}_{s_0}.
\end{equation}
Denote by $\bar P$ the mean-field (two-block averaged) transition matrix from the expected graph, whose block-averaged version is
\begin{equation}
    M =
    \begin{pmatrix}
        p_{\mathrm{in}} & p_{\mathrm{out}}\\
        p_{\mathrm{out}} & p_{\mathrm{out}}
    \end{pmatrix}.
\end{equation}
We write $\mu_{\mathrm{in}}$ and $\mu_{\mathrm{out}}$ for the mean-field average PPR per node inside and outside the planted group:
\begin{equation}
    \label{mu_in_and_out}
    \mu_{\mathrm{in}} =\frac{(1-\alpha)}{s}\big[(I-\alpha M)^{-1}\big]_{11}, \qquad
    \mu_{\mathrm{out}} =\frac{(1-\alpha)}{n-s}\big[(I-\alpha M)^{-1}\big]_{21},
\end{equation}
and define the \emph{mean-field gap}: $\Delta_{\mathrm{mean}}=\mu_{\mathrm{in}}-\mu_{\mathrm{out}}$.

\begin{theorem}[Detectability of PPR under Planted Laundering Group]
\label{thm:detectability}
Let $G\sim\mathrm{SBM}$ $(n,s,p_{\mathrm{in}},p_{\mathrm{out}})$ with $p_{\mathrm{in}}>p_{\mathrm{out}}$ and assume that degrees concentrate around their expectations (i.e., Lemma~\ref{lem:concentration} holds). Let $\mathbf{r}$ be the Personalized PageRank vector seeded at a known malicious account $s_0\in S^\star$ with a constant $\alpha\in(0,1)$, and $\varepsilon_{\infty}$ denote the entrywise perturbation scale arising from stochastic fluctuations of $P$ around the mean-field transition matrix $\bar P$. There exist constants $C,c>0$ such that, with probability at least $1-n^{-c}$, if the mean-field gap satisfies
\[
\Delta_{\mathrm{mean}} \;\ge\; C\cdot\varepsilon_{\infty},
\]
then 
\begin{enumerate}
    \item the average PPR score on $S^\star$ exceeds that on $V\setminus S^\star$ by at least $\tfrac{1}{2}\Delta_{\mathrm{mean}}$;
    \item ordering nodes by normalized PPR $r(v)/d(v)$ and selecting the prefix with smallest conductance recovers a subset $\widehat S$ with $|\widehat S\cap S^\star|\ge\gamma s$ for some constant $\gamma\in(0,1)$.
\end{enumerate}
In the standard regime $p_{\mathrm{out}}\gtrsim \log n/n$, the detectability condition simplifies to
\[
s\,(p_{\mathrm{in}}-p_{\mathrm{out}}) \;\gtrsim\; C'\sqrt{n p_{\max}\log n},
\]
where $p_{\max}=\max\{p_{\mathrm{in}},p_{\mathrm{out}}\}$. 
\end{theorem}

\vspace{0.5em}
\noindent\textit{Proof.}
See Appendix~\ref{app:proofs}.

Theorem~\ref{thm:detectability} establishes a detectability condition for identifying a money-laundering group using Personalized PageRank (PPR). Intuitively, PPR captures money-laundering behavior because it \emph{propagates the suspiciousness signal} from a known malicious seed node through the transactional network. Nodes densely connected to the seed, likely participating in group laundering, accumulate higher PPR scores, while weakly connected nodes receive lower scores. The theorem formalizes this by showing that a well-separated group will stand out in the PPR ranking.
In our proposed cross-bank PPR, we modify the transition matrix in PPR by assigning higher weights to edges more likely associated with laundering transactions as defined in Equation~\ref{PPR_weight}, which further amplifies scores for nodes strongly connected via suspicious interactions. Specifically, weighted edges increase the transition probability along suspicious links, causing PPR to concentrate more on nodes involved in laundering. This effectively increases \(\mu_{\mathrm{in}}\) relative to \(\mu_{\mathrm{out}}\), enlarging the gap \(\Delta_{\mathrm{mean}}\). Consequently, the detectability threshold is easier to satisfy, and PPR can recover a larger fraction of the laundering group or rank them higher in the graph cut procedure. On the other hand, to ensure that PPR starts with the malicious seed node, we reduce the probability of starting with a normal seed nodes by incorporating the cross-bank suspicious signal from other countries.

\subsection{Label Refinement via Label Propagation}
In the previous subsections, we demonstrated how cross-bank PPR captures group-wise money laundering patterns and provided theoretical insights into the detectability of PPR under planted laundering groups. In Section~\ref{PPR_interpretation}, we further visualize the identified laundering groups to highlight the interpretability of the proposed method. Beyond interpretation, a natural follow-up question arises: \emph{Can the discovered group-wise laundering patterns be further exploited to improve detection performance?} To address this question, we propose a label refinement strategy based on label propagation. This design is motivated by the observation that money laundering activities are inherently group-based and densely connected. Similar to the diffusion mechanism underlying PPR, we hypothesize that malicious signals identified within high-confidence laundering groups can be propagated to neighboring nodes. This is particularly useful for suspicious nodes that receive low confidence scores from the base detection module and may otherwise be overlooked. By propagating malicious signals from confidently identified nodes to their neighbors, we increase the likelihood of recovering such borderline cases.

Formally, given a weighted graph represented by a sparse adjacency matrix, we first construct a row-normalized transition matrix $\mathbf{S}$. Starting from the edge-level prediction $\hat{\mathbf{y}}_{uv}$ produced by the detection module, we compute a node-level malicious score for each account $u$ as
\begin{align}
    \hat{\mathbf{y}}_u = \frac{1}{|\{v \mid (v,u)\in E\}|} \sum_{v:(v,u)\in E} \hat{\mathbf{y}}_{vu},
\end{align}
which corresponds to the average malicious score over all edges incident to node $u$. These node-level scores form the initial label vector $\hat{\mathbf{y}}_{\text{nodes}}$.
Following the same idea of PPR, the label propagation process then iteratively updates the soft node labels according to
\begin{align}
    \mathbf{\hat{R}}^{(t+1)} &= \alpha \mathbf{S}\mathbf{\hat{R}}^{(t)} + (1-\alpha)\hat{\mathbf{y}}_{\text{nodes}}, \\
    \nonumber \mathbf{S} &= \mathbf{\hat{D}}^{-1}(A + A_{\text{PPR}}),
\end{align}
where $\mathbf{\hat{D}}$ is the degree matrix, $A$ denotes the original graph adjacency matrix, and $A_{\text{PPR}}$ encodes the group-wise laundering structures identified by cross-bank PPR. The parameter $\alpha \in (0,1)$ controls the strength of propagation. Iterations continue until convergence, measured by the $\ell_1$-norm difference between successive estimates, or until a predefined maximum number of iterations is reached. The resulting node labels are normalized to ensure numerical stability and rescale the malicious score within the range of [0,1]:
\begin{align}
    \mathbf{\hat{R}}^{(t+1)} \leftarrow \frac{\mathbf{\hat{R}}^{(t+1)}}{\max(\mathbf{\hat{R}}^{(t+1)})}.
\end{align}

Since $\mathbf{\hat{R}}$ represents node-level malicious scores, we need to convert them back to edge-level signals by assigning each edge $(u,v)$ the maximum score of its incident nodes,
\begin{align}
    \mathbf{\bar{R}}_{uv} = \max\left(\mathbf{\hat{R}}^{(t+1)}_u, \mathbf{\hat{R}}^{(t+1)}_v\right).
\end{align}
Finally, the original edge prediction $\hat{\mathbf{y}}_{uv}$ is refined by incorporating the propagated malicious signal:
\begin{align}
    \hat{\mathbf{y}}_{uv} \leftarrow (1-\alpha_{\text{lp}})\hat{\mathbf{y}}_{uv} + \alpha_{\text{lp}}\mathbf{\bar{R}}_{uv},
\end{align}
where $\alpha_{\text{lp}}$ balances the contribution between the base detector and the propagated group-level information.

\section{Hierarchical Adaptive Decision-Making for Financial Security}
\subsection{Hierarchical Adaptive Decision-Making module}
Regulatory authorities face clear-cut cases where intervention is straightforward: transactions flagged with very high confidence as suspicious should be immediately frozen. However, a critical challenge arises with transactions that receive low or moderate suspicion scores. Simply ignoring them risks missing coordinated laundering activities that unfold across multiple accounts and institutions, while indiscriminately freezing low-confidence transactions can lead to excessive disruption and unnecessary false positives. To address this challenge, we propose a \emph{Hierarchical Adaptive Decision-Making module} that integrates local threshold-based reinforcement learning with a global coordinator to ensure the consistency between local and global policies. The module treats the outputs of the graph-based detection as initial probabilistic labels for each transaction, providing prior knowledge of suspicious activity. Then, it captures both short-term (local bank-level) decision feedback and long-term (global system-level) optimization of cooperative anti-laundering policies.

Each bank $b$ manages a local transaction network $\mathcal{G}^b=(\mathcal{V}^b,\mathcal{E}^b)$, where each transaction $\mathcal{E}_{uv}^b$ carries an anomaly probability $\hat{\mathbf{y}}_{uv}$.  
At time $t$, the bank maintains an intervention threshold $\tau_b(t)$ and decides whether to intervene:
\begin{equation}
    a_{uv}^b:= \text{Intervene on } \mathcal{E}_{uv}^b \iff \hat{\mathbf{y}}_{uv} \geq \tau_b(t).
\end{equation}

To adaptively determine the threshold, we formulate the problem of adjusting thresholds as a reinforcement learning (RL) problem~\parencite{li2017deep, szepesvari2022algorithms, ladosz2022exploration}. At time $t$, the agent observes state $s_{uv}^b$, selects an action $a_{uv}^b$ corresponding to a threshold $\tau_b(t)$, and receives a reward $r_t^b$ that measures the quality of decisions across all edges $\mathcal{E}^{b}_t$ at the time $t$:
\begin{align}
    r_t^b &= \sum_{(u,v)\in \mathcal{E}^{b}_t} R(s_{uv}^b, a_{uv}^b, C_{uv}, y_{uv}) \\
    R(s_{uv}^b, a_{uv}^b, C_{uv}, y_{uv}) &= 
    \begin{cases}
        \alpha_{r1}\log(C_{uv}), & \text{if } y_{uv}=1 \text{ and } a_{uv}^b = \text{Freeze}   \\
        \alpha_{r2}\log(C_{uv}), & \text{if } y_{uv}=1 \text{ and } a_{uv}^b = \text{Monitor} \\
        -\alpha_{r3}\log(C_{uv}), & \text{if } y_{uv}=1 \text{ and } a_{uv}^b = \text{No Intervention} \\
        -\alpha_{r4}, & \text{if } y_{uv}=0 \text{ and } a_{uv}^b = \text{Freeze}   \\
        -\alpha_{r5}, & \text{if } y_{uv}=0 \text{ and } a_{uv}^b = \text{Monitor} \\
        \alpha_{r6}, & \text{if } y_{uv}=0 \text{ and } a_{uv}^b = \text{No Intervention} \\
     \end{cases}
\end{align}
where $\text{C}_{uv}$ is the cost of missing effective actions that should be deployed (which could be the amount of financial loss in this transaction), $y_{uv}$ is the label showing whether this transaction is illicit, $\alpha_{r1}$, $\alpha_{r2}$, $\alpha_{r3}$, $\alpha_{r4}$, $\alpha_{r5}$ and $\alpha_{r6}$ are positive hyperparameters balancing the importance of different actions in six scenarios. Specifically, when a transaction is malicious (i.e.,  $(y_{uv}=1$), we reward the model if the model begins to intervene in this transaction either by freezing or monitoring this transaction, with $\alpha_{r1} > \alpha_{r2} > \alpha_{r3}$ showing the different weight for three actions. Otherwise, we reduce its rewards if no intervention is deployed. When a transaction is normal activity, we reward the model if it does not intervene in this normal behavior. Otherwise, we penalize the model at different degrees with $\alpha_{r4} \geq \alpha_{r6} >\alpha_{r5}$. This reward design encourages high-probability suspicious transactions to trigger appropriate interventions, penalizes unnecessary actions on benign transactions, and incorporates operational cost through both the constant penalties and the cost-aware logarithmic terms. The goal of the local policy $\pi_\theta^b(a|s_{uv}^b)$ is to maximize the expected cumulative reward:
\begin{equation}
    J(\pi_\theta^b) = \mathbb{E}_{\pi^b}\!\Big[\sum_{t=0}^{\infty} \gamma^t r_t^b \Big].
\end{equation}

Our dynamic decision-making framework draws on the tradition of modeling sequential choice under uncertainty, where agents make decisions that affect both immediate outcomes and future states~\parencite{mehta2017ni}. To address the negative impact of bank isolation issue, we propose a global-local coordinator to ensure the consistency between local and global policies. Specifically, local institutions adaptively respond to evolving transaction patterns, while the global coordinator aligns their behavior toward collective objectives. The global policy $\pi_\phi^{\text{global}}$ updates local thresholds using soft coordination:
\begin{equation}
    \tau_b(t+1) \leftarrow \tau_b(t) + \eta_g \big(\bar{\tau}(t) - \tau_b(t)\big), \quad 
    \bar{\tau}(t) = \sum_{b=1}^{B} w_b \tau_b(t), 
\end{equation}
where $\omega_b$ represents the relative institutional importance or transaction volume weight. To ensure the consistency between local and global policies, we define the coupling constraint as:
\begin{align}
    \mathcal{L}_{\text{couple}} &= 
    \sum_{b=1}^B \big\| \tau_b(t) - \bar{\tau}(t) \big\|^2
    + \xi \sum_{b=1}^B D_{\mathrm{KL}}\!\big(\pi_\theta^b \,||\, \pi_\phi^{\text{global}}\big), 
\end{align}
where the first term enforces consistency between local and global thresholds, and the second term aligns local policies with global intent through a KL-divergence regularization weighted by $\xi$. 
The joint optimization objective becomes:
\begin{equation}
    \max_{\{\pi_\theta^b\},\, \pi_\phi^{\text{global}}}
    \sum_{b=1}^B J(\pi_\theta^b)
    - \lambda_{\text{c}} \mathcal{L}_{\text{couple}},
\end{equation}
where $\lambda_{\text{c}}$ controls the strength of policy coupling.
The proposed hierarchical adaptive decision-making framework establishes a principled coordination structure for multi-institution financial security systems. Local RL agents autonomously learn transaction-level thresholds for each bank, while a global coordinator ensures cooperative, privacy-preserving, and regulation-compliant adaptation.

\section{Empirical Results}

\subsection{Dataset Statistics and Feature Preprocessing}
Our empirical analysis leverages the IBM Anti–Money Laundering (AML) Small dataset, which comprises over 5 million transaction records collected between September 1 and September 10, 2022. IBM Anti–Money Laundering dataset is a synthetic dataset, which is generated through a structured, multi-stage simulation framework designed to mimic real-world financial behavior while injecting realistic illicit activities~\parencite{altman2023realistic}. According to Altman et al.~\parencite*{altman2023realistic}, the process begins by constructing a population of synthetic customers, accounts, and financial entities whose demographic and behavioral attributes are sampled from empirically observed distributions. A transaction network is then formed by modeling normal financial activities using probabilistic rules calibrated to real banking data, capturing patterns such as salary deposits, bill payments, peer-to-peer transfers, and business-to-consumer flows. On top of this baseline, the system overlays money-laundering schemes—such as structuring, smurfing, round-tripping, and funnel accounts—through explicit scenario scripts that specify how illicit funds are introduced, layered, and integrated. Each scenario defines the actors involved, temporal patterns, transaction amounts, and network structures, ensuring that both benign and suspicious behaviors emerge organically within the same simulated environment. This combination of bottom-up population modeling and top-down illicit activity injection allows the IBM AML dataset to faithfully approximate real transaction ecosystems while providing ground-truth labels for evaluating detection modus~\footnote{Please refer to the paper by ~\parencite{altman2023realistic} for more details of generating the AML dataset.}.

\begin{table}[h!]
\centering
\caption{Statistics of transaction by country}
\resizebox{0.9\textwidth}{!}{
\begin{tabular}{lcccc}
\hline
\textbf{Country} & \textbf{\#Accounts} & \textbf{\#Total Transactions} & \textbf{\#Illicit Transactions} & \textbf{Ratio of Illicity}\\
\hline
United States & 71,796 & 855,006  & 3,043 & 0.35\%\\
Germany       & 31,566 & 275,129  & 1,308 & 0.47\%\\
France        & 28,126 & 244,589  & 952   & 0.38\%\\
Italy         & 23,262 & 194,155  & 804   & 0.41\%\\
Spain         & 24,363 & 207,098  & 763   & 0.36\%\\
China         & 21,345 & 181,341  & 961   & 0.52\%\\
Rest Countries & 16,213 & 111,847 & 952   & 0.85\%\\
\hline
\end{tabular}}
\label{tab:country_graph_stats}
\end{table}

Since most of the baseline methods is not scalable to the large-scale dataset, we first compare the performance of our detection module with the existing methods in a subset with 1.4 million transaction records in this section and then we demonstrate the scalability of our method in Section~\ref{scalability_analysis}. We summarize the transaction networks across different countries in Table~\ref{tab:country_graph_stats} for the subset with 1.4 million transaction records. To mimic the cross-boarder money-laundering detection, we split the entire dataset into multiple subsets based on the nationality of banks. Transactions involving multiple banks are duplicated so that they appear in the networks of all relevant countries. After splitting, the United States represents the largest network, with over 71,000 accounts and approximately 855,000 transactions, followed by major European economies such as Germany, France, Italy, and Spain. Despite the variations in network size, a consistent pattern emerges across all markets: illicit transactions account for less than 1\% of total activity. For instance, while China exhibits the highest proportion of suspicious transactions at 0.52\%, the overall prevalence of illicit activity remains extremely low, reflecting the severe class imbalance that characterizes real-world financial data. Interestingly, the “Rest Countries” category, comprising smaller or less-regulated markets, shows a noticeably higher illicit ratio of 0.85\%, suggesting that money-laundering activities may be more concentrated in regions with weaker oversight or fragmented compliance mechanisms. The prevalence of informal financial channels in emerging markets~\parencite{bao2018nisingh} creates additional vulnerabilities that sophisticated laundering networks can exploit. This statistical profile underscores the operational challenge faced by financial institutions: effectively identifying rare, high-risk activities within overwhelmingly legitimate transaction flows. We partition each country-level dataset into a 5\% training set and a 95\% test set to evaluate model generalization under limited supervised data. 

On the IBM AML dataset, a critical feature is \emph{transaction amount}, which ranges from 1 cent to 8.04 billion dollars. Such a wide scale increases the number of training iterations required for most models to converge, and prior work has shown that appropriate feature normalization significantly accelerates deep network training~(\cite{ioffe2015batch}). Therefore, feature normalization is essential in our preprocessing pipeline. During preprocessing, we observed that the \emph{order} of feature normalization and data partitioning meaningfully affects the performance of both our method and baseline models. This performance degradation stems from \emph{subpopulation shift}~\parencite{koh2021wilds}---a form of distribution shift where the overall population remains fixed, but the characteristics or proportions of underlying subgroups differ. In our setting, these subpopulations correspond to transactions originating from different countries/markets. When features are normalized \emph{before} partitioning, each country receives data with similar statistical properties. However, when we first split the data by country and then apply feature normalization within each subset, the resulting country-level scales diverge substantially, increasing the degree of subpopulation shift.

\begin{table}[ht]
\caption{Data statistics of the \emph{transaction amount} feature across countries. For each country, $\Delta_{\text{metric}}$ ($\text{metric} \in {\text{Min}, \text{Max}, \text{Mean}, \text{SD}}$) denotes the difference between the country-level statistic and the global-level (Overall) statistic. Large deviations across these metrics illustrate the presence of subpopulation shift after data partition.}
\centering
\small
\resizebox{0.9\textwidth}{!}{
\begin{tabular}{lcccccccc}
\hline
\textbf{GraphName} & \textbf{Min} & \textbf{Max} & \textbf{Mean} & \textbf{SD} & \textbf{$\Delta$ Min} & \textbf{$\Delta$ Max} & \textbf{$\Delta$ Mean} & \textbf{$\Delta$ SD} \\
\hline
United States       & 0.01 & 2134359601 & 388627.61 & 10616792.77 & 0 & 5911955518 & 50279.46 & 9334910.41 \\
Germany             & 0.01 & 8046315118 & 541468.28 & 31701228    & 0 & 0          & -102561.21 & -11749524.82 \\
France              & 0.02 & 2134359601 & 410596.91 & 11738295.85 & -0.01 & 5911955518 & 28310.16 & 8213407.33 \\
Italy               & 0.01 & 1825924651 & 446492.93 & 11949073    & 0 & 6220390468 & -7585.86 & 8002630.18 \\
Spain               & 0.01 & 1825924651 & 472743.68 & 11263536.13 & 0 & 6220390468 & -33836.61 & 8688167.05 \\
China               & 0.01 & 1825924651 & 432847.84 & 11948125.86 & 0 & 6220390468 & 6059.23 & 8003577.32 \\
Rest Countries      & 0.01 & 7512426017 & 618635.73 & 34107069.54 & 0 & 533889100.7 & -179728.66 & -14155366.36 \\
Overall             & 0.01 & 8046315118 & 438907.07 & 19951703.18 & 0 & 0          & 0 & 0 \\
\hline
\end{tabular}}
\label{data_distribution_shift_demo}
\end{table}

Subpopulation shift poses a practical challenge: (1) data cannot be shared across banks due to privacy and regulatory constraints, yet (2) proper normalization is necessary for model convergence. To validate our hypothesis that post-partition normalization amplifies subpopulation shift, we conduct an exploratory analysis using Min--Max normalization:
\begin{equation}
    \label{min_max_normalization}
    x^{\text{scaled}}_i = \frac{x_i - x^{\text{min}}_i}{x^{\text{max}}_i - x^{\text{min}}_i}.
\end{equation}
Table~\ref{data_distribution_shift_demo} reports the minimum, maximum, mean, and standard deviation (denoted as SD) of \emph{transaction amount} for each country and for the global dataset. To quantify deviations between country-level and global distributions, we compute
\begin{equation}
    \Delta_{\text{metric}} = \text{Overall}_{\text{metric}} - \text{Country}_{\text{metric}},
\end{equation}
where $\text{metric} \in \{\text{Min}, \text{Max}, \text{Mean}, \text{SD}\}$. Across all countries, the last three metrics vary substantially. This implies that widely used normalization strategies, such as Min--Max Normalization, Z-Score Normalization, and Mean Normalization, are likely to induce significant subpopulation shift when country-level statistics are used as proxies for global ones.

Next, we empirically evaluate how subpopulation shift affects model performance. To isolate its impact, we design a controlled experiment with two settings. \textbf{(1) Country-level Normalization:} We normalize the feature in each subset independently based on the country-level statistics. Notice that the country-level normalization induces the significant subpopulation shift due to the significant difference between the country-level statistics and global statistics. \textbf{(2) Global-level Normalization:} We normalize the feature in the entire dataset using global statistics. Because the second setting applies consistent scaling across all countries, it eliminates subpopulation shift. In both settings, we aggregate all transactions from different subsets into a single dataset and train each model once on the unified dataset. Table~\ref{subpopulation_shift_results} presents the results, where Panel~A corresponds to the setting country-level normalization and Panel~B to the setting with global-level normalization. The results reveal two key findings. First, Random Forest shows mixed behavior: with global-level normalization, its Type~II error decreases by over 11\%, but its AUPRC drops by about 3\% compared with country-level normalization, making it difficult to infer the overall effect solely from this model. Second, rest models, including Support Vector Machine (SVM), logistic regression, Federated MLP, and our proposed method, achieve higher AUPRC and lower Type II error when subpopulation shift is removed or replacing country-level normalization with global-level normalization. These empirical findings confirm that subpopulation shift meaningfully reduces the detection performance of most machine learning models in the AML setting.

\begin{table*}[ht!]
\centering
\caption{Comparison of Methods Under \emph{Country-level Normalization} (Panel A) and \emph{Global-level Normalization} (Panel B), where \emph{Country-level Normalization} induces significant subpopulation shift.}
\small
\resizebox{0.9\textwidth}{!}{
\begin{tabular}{l|cc|cc|cc|cc|cc}
\multicolumn{11}{c}{\textbf{Panel A: Country-level Normalization}} \\
\hline
\textbf{Graph} &
\multicolumn{2}{c|}{\textbf{Logistic Regression}} &
\multicolumn{2}{c|}{\textbf{SVM}} &
\multicolumn{2}{c|}{\textbf{Random Forest}} &
\multicolumn{2}{c|}{\textbf{Federated MLP}} &
\multicolumn{2}{c}{\textbf{Our Method}} \\
& AUPRC  & Type II & AUPRC  & Type II & AUPRC  & Type II & AUPRC  & Type II & AUPRC  & Type II \\
\hline
United States        & 0.2381 & 0.0193 & 0.3412 & 1.0000 & 0.7503 & 0.3961 & 0.4977 & 0.0699 & 0.6881 & 0.0636 \\
Germany   & 0.2875 & 0.0293 & 0.3340 & 1.0000 & 0.4101 & 0.7594 & 0.3771 & 0.0920 & 0.5544 & 0.0657 \\
France    & 0.2370 & 0.0179 & 0.3086 & 1.0000 & 0.7266 & 0.3824 & 0.3369 & 0.0702 & 0.5299 & 0.0413 \\
Italy     & 0.2717 & 0.0147 & 0.3070 & 1.0000 & 0.6896 & 0.4517 & 0.3759 & 0.1195 & 0.5363 & 0.0753 \\
Spain     & 0.2712 & 0.0157 & 0.2754 & 1.0000 & 0.7010 & 0.3916 & 0.3347 & 0.0892 & 0.5026 & 0.0507 \\
China     & 0.3202 & 0.0217 & 0.3443 & 1.0000 & 0.7047 & 0.4393 & 0.3924 & 0.0564 & 0.4716 & 0.0477 \\
Rest Countries & 0.5487 & 0.0156 & 0.6186 & 1.0000 & 0.6542 & 0.7337 & 0.6527 & 0.1020 & 0.7854 & 0.0666 \\
\hline
Overall        & 0.3106 & 0.0192 & 0.3613 & 1.0000 & \textbf{0.6624} & 0.5077 & 0.4239 & 0.0856 & 0.6046 & 0.0600 \\
\hline
\end{tabular}}
\resizebox{0.9\textwidth}{!}{
\begin{tabular}{l|cc|cc|cc|cc|cc}
\multicolumn{11}{c}{\textbf{Panel B: Global-level Normalization}} \\
\hline
\textbf{Graph} &
\multicolumn{2}{c|}{\textbf{Logistic Regression}} &
\multicolumn{2}{c|}{\textbf{SVM}} &
\multicolumn{2}{c|}{\textbf{Random Forest}} &
\multicolumn{2}{c|}{\textbf{Federated MLP}} &
\multicolumn{2}{c}{\textbf{Our Method}} \\
& AUPRC  & Type II & AUPRC  & Type II & AUPRC  & Type II & AUPRC  & Type II & AUPRC  & Type II \\
\hline
United States   & 0.2761 & 0.0076 & 0.3752 & 0.8974 & 0.6446 & 0.3508 & 0.5136 & 0.0542 & 0.6922 & 0.0641 \\
Germany         & 0.2886 & 0.0020 & 0.3642 & 0.8443 & 0.5961 & 0.4014 & 0.3867 & 0.0748 & 0.5617 & 0.0647 \\
France          & 0.2471 & 0.0041 & 0.3453 & 0.8308 & 0.5714 & 0.4017 & 0.3485 & 0.0495 & 0.5399 & 0.0440 \\
Italy           & 0.2722 & 0.0033 & 0.3801 & 0.8003 & 0.5758 & 0.4354 & 0.3757 & 0.1015 & 0.5406 & 0.0736 \\
Spain           & 0.2585 & 0.0017 & 0.3300 & 0.8269 & 0.5762 & 0.3881 & 0.3379 & 0.0752 & 0.5153 & 0.0507 \\
China           & 0.3017 & 0.0043 & 0.3458 & 0.8699 & 0.5919 & 0.4118 & 0.3865 & 0.0477 & 0.4774 & 0.0477 \\
Rest Countries  & 0.5742 & 0.0042 & 0.6151 & 0.8229 & 0.7954 & 0.3541 & 0.6804 & 0.0949 & 0.7833 & 0.0609 \\
\hline
Overall         & \textbf{0.3169} & \textbf{0.0039} & \textbf{0.3937} & \textbf{0.8418} & 0.6216 & \textbf{0.3919} & \textbf{0.4327} & \textbf{0.0711} & \textbf{0.6105} & \textbf{0.0596} \\
\hline
\end{tabular}}
\label{subpopulation_shift_results}
\end{table*}

Although our method with global-level normalization outperforms the country-level normalization baseline, privacy concerns prevent institutions from sharing raw data across domains, making global-level normalization less practical in real-world settings. Moreover, the evolving nature of transaction data necessitates frequent updates to normalization statistics, further complicating coordination across institutions. To address this issue, we propose a simple yet effective solution, termed fixed-value min–max normalization. Unlike the min–max normalization defined in Equation~(\ref{min_max_normalization}), which relies on global-level statistics computed over the entire dataset, our approach randomly selects the minimum and maximum values from a country-level subset and shares only these fixed values with other subsets. As a result, all subsets apply the same minimum and maximum values during normalization without directly sharing sensitive data. To evaluate whether it alleviate subpopulation shift, we conduct a comparative empirical study to examine whether the proposed fixed-value min–max normalization can achieve performance comparable to global-level normalization. The results in Table~\ref{fixed_min_max_results} show that our method with fixed-value min–max normalization achieves slightly better performance than global-level min–max normalization across all four evaluation metrics, demonstrating the effectiveness of the proposed normalization strategy. Overall, this fixed-value min–max normalization strategy offers two key advantages. First, it explicitly addresses privacy concerns by avoiding the exchange of raw data across institutions. Second, it effectively mitigates subpopulation shift by enforcing a consistent normalization scale across different subsets.

\begin{table}[ht!]
    \centering
    \caption{Performance comparison of our method under different normalization strategies (Global-level Min-max Normalization vs Fixed Value Min-max Normalization. }
    \resizebox{0.90\textwidth}{!}{
    \begin{tabular}{l|cccc|cccc}
    \hline
    \multirow{2}{*}{\textbf{Market}} 
    & \multicolumn{4}{c|}{\textbf{Global-level Min--max Normalization}} 
    & \multicolumn{4}{c}{\textbf{Fixed Value Min--max Normalization}} \\
    \cline{2-9}
    & AUCROC & AUPRC & Type I Error & Type II Error 
    & AUCROC & AUPRC & Type I Error & Type II Error \\
    \hline
    United States           & 0.9818 & 0.6922 & 0.0520 & 0.0641 & 0.9831 & 0.7004 & 0.0524 & 0.0660 \\
    Germany                 & 0.9773 & 0.5617 & 0.0723 & 0.0647 & 0.9778 & 0.5880 & 0.0702 & 0.0512 \\
    France                  & 0.9787 & 0.5399 & 0.0719 & 0.0440 & 0.9787 & 0.5515 & 0.0711 & 0.0480 \\
    Italy                   & 0.9778 & 0.5406 & 0.0670 & 0.0736 & 0.9789 & 0.5730 & 0.0656 & 0.0626 \\
    Spain                   & 0.9793 & 0.5153 & 0.0680 & 0.0507 & 0.9790 & 0.5048 & 0.0653 & 0.0452 \\
    China                   & 0.9720 & 0.4774 & 0.0710 & 0.0477 & 0.9740 & 0.4955 & 0.0687 & 0.0358 \\
    Rest Countries          & 0.9865 & 0.7833 & 0.0396 & 0.0609 & 0.9866 & 0.7775 & 0.0390 & 0.0558 \\
    \hline
    Overall                 & 0.9800 & 0.6105 & 0.0610 & 0.0596 & \textbf{0.9809} & \textbf{0.6121} & \textbf{0.0602} & \textbf{0.0553} \\
    \hline
    \end{tabular}
    \label{fixed_min_max_results}
    }
\end{table}

Next, we examine the effect of data fragmentation on model performance. In this experiment, we first normalize the features and then partition the dataset. In the first setting, machine learning models are trained separately on each country-level subset. In the second setting, we aggregate all transactions into a single dataset and train each model once on the unified data. Comparing these two settings allows us to test our hypothesis that data fragmentation negatively affects model performance. Table~\ref{data_fragmentation_results} reports the results, with Panel~A corresponding to the fragmented (country-level) setting and Panel~B to the aggregated setting. The findings show that Logistic Regression, SVM, and the MLP with Focal Loss (MLP+Focal) consistently achieve higher AUPRC scores when trained on the aggregated dataset. Notably, SVM, Logistic Regression and MLP + Focal improve AUPRC and reduce Type~II error in the aggregated setting. In particular, MLP + Focal increases AUPRC by more than 4.5\% and reduces Type~II error from 20\% to 6.18\%. These results demonstrate that data fragmentation leads to degraded performance relative to training on the fully aggregated dataset.

\begin{table*}[ht]
\centering
\caption{Examining the Impact of Data Fragmentation.}

\resizebox{0.9\textwidth}{!}{
\begin{tabular}{lcccccccc}
\multicolumn{9}{c}{\textbf{Panel A: Results of Methods \textbf{With Partition}}} \\
\hline
\multirow{2}{*}{\textbf{Market}} 
& \multicolumn{2}{c}{\textbf{Logistic Regression}} 
& \multicolumn{2}{c}{\textbf{SVM}} 
& \multicolumn{2}{c}{\textbf{Random Forest}}
& \multicolumn{2}{c}{\textbf{MLP+Focal}} \\
\cline{2-9}
& AUPRC & Type II & AUPRC & Type II & AUPRC & Type II & AUPRC & Type II \\
\hline
United States        & 0.2262 & 0.0076 & 0.3102 & 1.0000 & 0.6173 & 0.4350 & 0.5012 & 0.1568 \\
Germany   & 0.2882 & 0.0020 & 0.3222 & 0.9990 & 0.6176 & 0.5976 & 0.3991 & 0.2042 \\
France    & 0.2495 & 0.0041 & 0.2782 & 1.0000 & 0.6689 & 0.6589 & 0.2916 & 0.2545 \\
Italy     & 0.2636 & 0.0033 & 0.3266 & 0.9411 & 0.6693 & 0.6301 & 0.3569 & 0.2815 \\
Spain     & 0.2601 & 0.0017 & 0.2234 & 1.0000 & 0.6232 & 0.6731 & 0.3351 & 0.2028 \\
China     & 0.3119 & 0.0564 & 0.3055 & 1.0000 & 0.5528 & 0.6922 & 0.3249 & 0.2095 \\
Rest Countries & 0.5198 & 0.0042 & 0.5626 & 0.1261 & 0.8133 & 0.2082 & 0.6120 & 0.0907 \\
Overall        & 0.3027 & 0.0113 & 0.3327 & 0.8666 & \textbf{0.6518} & 0.5564 & 0.4030 & 0.2000 \\
\hline
\end{tabular}
}

\resizebox{0.9\textwidth}{!}{
\begin{tabular}{lcccccccc}
\multicolumn{9}{c}{\textbf{Panel B: Results of Methods \textbf{Without Partition}}} \\
\hline
\multirow{2}{*}{\textbf{Market}}
& \multicolumn{2}{c}{\textbf{Logistic Regression}} 
& \multicolumn{2}{c}{\textbf{SVM}} 
& \multicolumn{2}{c}{\textbf{Random Forest}}
& \multicolumn{2}{c}{\textbf{MLP+Focal}} \\
\cline{2-9}
& AUPRC & Type II & AUPRC & Type II & AUPRC & Type II & AUPRC & Type II \\
\hline
United States   & 0.2761 & 0.0076 & 0.3752 & 0.8974 & 0.6446 & 0.3508 & 0.4743 & 0.0623 \\
Germany         & 0.2886 & 0.0020 & 0.3642 & 0.8443 & 0.5961 & 0.4014 & 0.4081 & 0.0688 \\
France          & 0.2471 & 0.0041 & 0.3453 & 0.8308 & 0.5714 & 0.4017 & 0.3860 & 0.0454 \\
Italy           & 0.2722 & 0.0033 & 0.3801 & 0.8003 & 0.5758 & 0.4354 & 0.4073 & 0.0769 \\
Spain           & 0.2585 & 0.0017 & 0.3300 & 0.8269 & 0.5762 & 0.3881 & 0.3685 & 0.0455 \\
China           & 0.3017 & 0.0043 & 0.3458 & 0.8699 & 0.5919 & 0.4118 & 0.4004 & 0.0506 \\
Rest Countries  & 0.5742 & 0.0042 & 0.6151 & 0.8229 & 0.7954 & 0.3541 & 0.6917 & 0.0836 \\
Overall         & \textbf{0.3169} & \textbf{0.0039} & \textbf{0.3937} & \textbf{0.8418} & 0.6216 & \textbf{0.3919} & \textbf{0.4480} & \textbf{0.0618} \\
\hline
\end{tabular}
}
\label{data_fragmentation_results}
\end{table*}

\subsection{Empirical Results of Detection Module}
\subsubsection{Effectiveness of Detection Module} 
To verify the capability of Focal loss addressing the severe class imbalance issue, we compare the model performance with the Cross-Entropy and Focal Loss approaches in Table~\ref{tab:different_loss_performance}. Both of them deliver strong overall predictive accuracy, as reflected in their comparable AUROC scores across all country-level transaction networks. However, the distinction between the two becomes more apparent when considering the model’s ability to detect rare money-laundering activities. The Cross-Entropy Loss tends to favor the majority class (normal transactions), leading to lower false alarm rates (Type I errors) but a higher tendency to overlook true laundering cases (Type II errors). In contrast, the Focal Loss introduces a more balanced learning process by assigning greater weight to these hard-to-detect, high-risk transactions. This results in a better trade-off between precision and recall, reducing missed detections while maintaining strong overall accuracy. This improvement is particularly valuable, as identifying a higher proportion of true laundering events enhances compliance effectiveness and reduces the risk of undetected illicit flows even at the cost of slightly more alerts. Essentially, the use of Focal Loss helps the model more effectively address the data imbalance challenge inherent in illicit financial activity detection, improving both operational robustness and decision confidence.

\begin{table}[h!]
    \centering
    \caption{Performance comparison with and without graph cut loss}
    \resizebox{0.9\textwidth}{!}{
    \begin{tabular}{l|cccc|cccc}
        \hline
        \multirow{2}{*}{\textbf{Market}} & \multicolumn{4}{c|}{\textbf{No graph cut loss}} & \multicolumn{4}{c}{\textbf{Add graph cut loss}} \\
        \cline{2-9}
         & AUCROC & AUPRC & Type I Error & Type II Error & AUCROC & AUPRC & Type I Error & Type II Error \\
        \hline
        United States       & 0.9821 & 0.7098 & 0.0488 & 0.0731 & 0.9831 & 0.7004 & 0.0524 & 0.0660 \\
        Germany             & 0.9782 & 0.5936 & 0.0670 & 0.0603 & 0.9778 & 0.5880 & 0.0702 & 0.0512 \\
        France              & 0.9783 & 0.5688 & 0.0677 & 0.0526 & 0.9787 & 0.5515 & 0.0711 & 0.0480 \\
        Italy               & 0.9783 & 0.5726 & 0.0639 & 0.0681 & 0.9789 & 0.5730 & 0.0656 & 0.0626 \\
        Spain               & 0.9800 & 0.5196 & 0.0614 & 0.0413 & 0.9790 & 0.5048 & 0.0653 & 0.0452 \\
        China               & 0.9740 & 0.5026 & 0.0648 & 0.0472 & 0.9740 & 0.4955 & 0.0687 & 0.0358 \\
        Rest Countries      & 0.9865 & 0.7723 & 0.0393 & 0.0510 & 0.9866 & 0.7775 & 0.0390 & 0.0558 \\
        \hline
        Overall             & 0.9806 & 0.6307 & 0.0570 & 0.0605 & \textbf{0.9809} & \textbf{0.6121} & \textbf{0.0602} & \textbf{0.0553} \\
        \hline
    \end{tabular}}
    \label{tab:graph_cut_comparison}
\end{table}

We further investigate the impact of incorporating the graph cut loss, which is designed to encourage clearer structural separation between high-risk and low-risk transaction regions in the learned representation space. As shown in Table~\ref{tab:graph_cut_comparison}, adding the graph cut loss yields modest but consistent improvements in overall AUCROC, accompanied by a reduction in Type II errors, indicating enhanced sensitivity to illicit activities. While AUPRC slightly decreases for several country-level subgraphs, the observed decline is generally accompanied by lower missed-detection rates, suggesting a shift toward more effective recall of positive cases. This trade-off reflects a more risk-aware detection strategy, where the model prioritizes identifying subtle laundering patterns embedded in complex transaction graphs. Overall, the results demonstrate that the graph cut loss serves as a complementary regularization mechanism that improves global discrimination performance and robustness, particularly in reducing false negatives across heterogeneous banking networks.

We also examine the effectiveness of virtual super-node. The virtual super-node is designed to overcome the long-standing challenge of bank isolation, where strict regulatory and privacy rules prevent financial institutions from directly sharing cross-border transaction data. Rather than exchanging raw information, the super-node enables banks to share structured relationship insights, serving as a bridge that connects otherwise separated financial subgraphs. This design allows the model to recognize patterns of cross-border transactions that would typically remain hidden within institutional boundaries. As shown in Table~\ref{tab:supernode_comparison}, introducing the virtual super-node leads to consistent performance gains across all country-level banking networks. The improvement in key performance indicators such as AUCROC and AUPRC demonstrates that incorporating structured interconnections via virtual super-node helps strengthen risk detection accuracy while reducing false positives. Although a slight rise in Type II errors suggests a more conservative detection approach, the overall gains indicate that the super-node provides a practical, privacy-conscious mechanism for enhancing collaborative anti-money-laundering intelligence across jurisdictions.

\begin{table}[h!]
    \centering
    \caption{Performance comparison with and without super-node}
    \resizebox{0.9\textwidth}{!}{
    \begin{tabular}{l|cccc|cccc}
        \hline
        \multirow{2}{*}{\textbf{Market}} & \multicolumn{4}{c|}{\textbf{Without super-node}} & \multicolumn{4}{c}{\textbf{With super-node}} \\
        \cline{2-9}
         & AUCROC & AUPRC & Type I Error & Type II Error & AUCROC & AUPRC & Type I Error & Type II Error \\
        \hline
        United States       & 0.9813 & 0.6747 & 0.0525 & 0.0605 & 0.9831 & 0.7004 & 0.0524 & 0.0660 \\
        Germany             & 0.9767 & 0.5579 & 0.0728 & 0.0617 & 0.9778 & 0.5880 & 0.0702 & 0.0512 \\
        France              & 0.9772 & 0.5047 & 0.0730 & 0.0440 & 0.9787 & 0.5515 & 0.0711 & 0.0480 \\
        Italy               & 0.9775 & 0.5385 & 0.0676 & 0.0704 & 0.9789 & 0.5730 & 0.0656 & 0.0626 \\
        Spain               & 0.9774 & 0.4992 & 0.0685 & 0.0490 & 0.9790 & 0.5048 & 0.0653 & 0.0452 \\
        China               & 0.9720 & 0.4626 & 0.0720 & 0.0448 & 0.9740 & 0.4955 & 0.0687 & 0.0358 \\
        Rest Countries      & 0.9872 & 0.7932 & 0.0403 & 0.0609 & 0.9866 & 0.7775 & 0.0390 & 0.0558 \\
        \hline
        Overall             & 0.9794 & 0.5948 & 0.0617 & 0.0571 & \textbf{0.9809} & \textbf{0.6121} & \textbf{0.0602} & \textbf{0.0553} \\
        \hline
    \end{tabular}}
    \label{tab:supernode_comparison}
\end{table}

\begin{table}[h!]
\centering
\caption{Comparison of model performance for different types of loss}
\resizebox{0.9\textwidth}{!}{
\begin{tabular}{l|cccc|cccc}
\hline
\multirow{2}{*}{\textbf{Market}} & \multicolumn{4}{c|}{\textbf{Cross-Entropy Loss}} & \multicolumn{4}{c}{\textbf{Focal Loss}} \\
\cline{2-9}
 & AUCROC & AUPRC & Type I Error & Type II Error & AUCROC & AUPRC & Type I Error & Type II Error \\
\hline
United States  & 0.9821 & 0.6973 & 0.0042 & 0.4987 & 0.9831 & 0.7004 & 0.0524 & 0.0660 \\
Germany        & 0.9769 & 0.5709 & 0.0079 & 0.5774 & 0.9778 & 0.5880 & 0.0702 & 0.0512 \\
France         & 0.9782 & 0.5443 & 0.0072 & 0.6259 & 0.9787 & 0.5515 & 0.0711 & 0.0480 \\
Italy          & 0.9779 & 0.5579 & 0.0062 & 0.6334 & 0.9789 & 0.5730 & 0.0656 & 0.0626 \\
Spain          & 0.9786 & 0.5261 & 0.0049 & 0.6434 & 0.9790 & 0.5048 & 0.0653 & 0.0452 \\
China          & 0.9727 & 0.4735 & 0.0075 & 0.7370 & 0.9740 & 0.4955 & 0.0687 & 0.0358 \\
Rest Countries & 0.9870 & 0.7957 & 0.0058 & 0.5680 & 0.9866 & 0.7775 & 0.0390 & 0.0558 \\
\hline
Overall        & 0.9800 & \textbf{0.6171} & \textbf{0.0056} & 0.5828 
               & \textbf{0.9809} & 0.6121 & 0.0602 & \textbf{0.0553} \\
\hline
\end{tabular}
}
\label{tab:different_loss_performance}
\end{table}

\begin{figure}[htbp]
    \centering
    \centering
    \includegraphics[width=\linewidth]{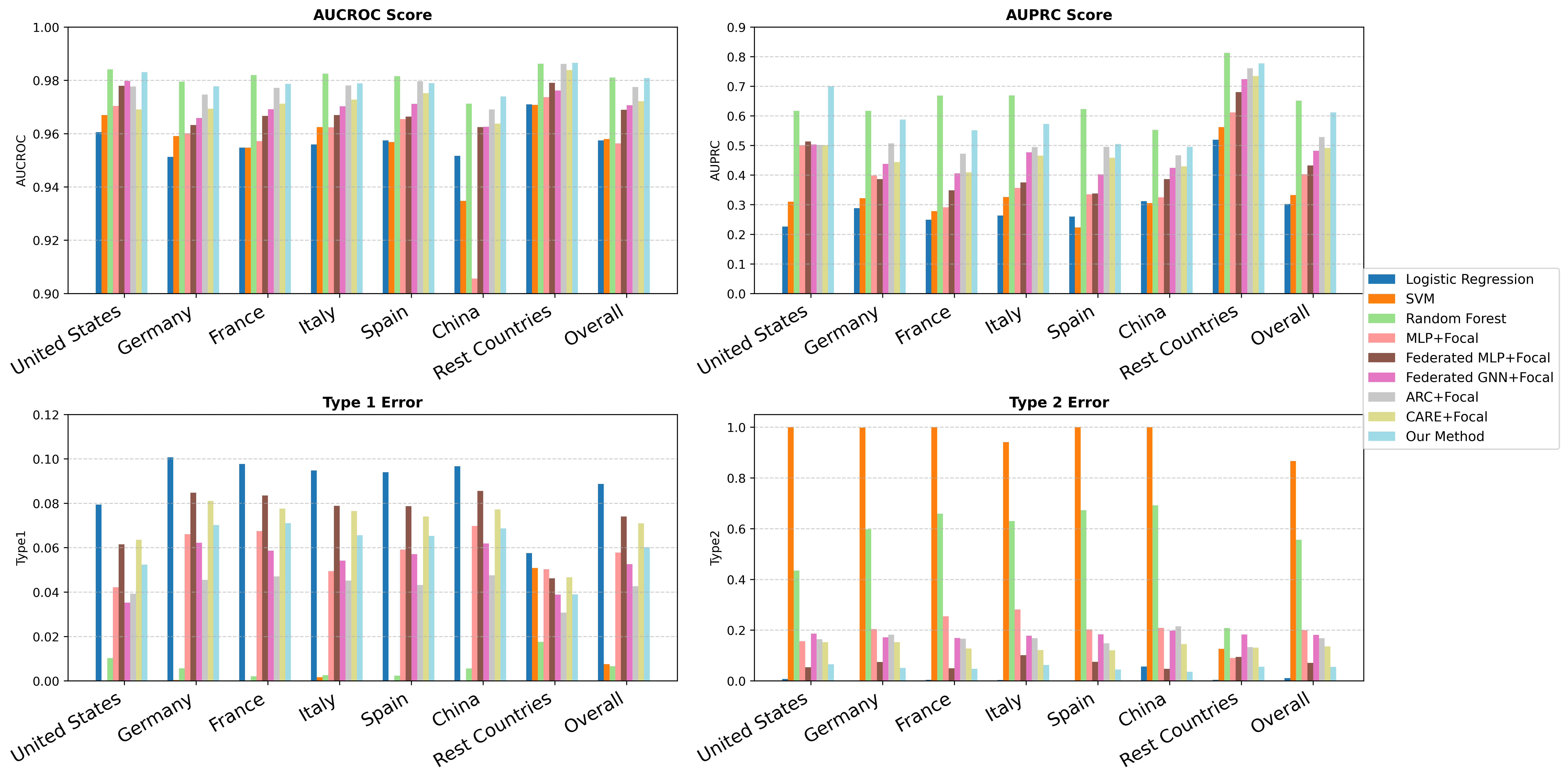}
    \caption{Comparison with traditional machine learning methods and deep learning approaches. All baseline methods excepted two federated methods are trained on every sub-dataset (i.e., each country) separately.}
    \label{fig:detection_model_results}
\end{figure}

 
To evaluate the effectiveness of our proposed method, we compare it against three traditional machine learning algorithms, (i.e., Logistic Regression, Support Vector Machine (SVM), and Random Forest), three deep learning baselines (i.e., a three-layer MLP with Focal Loss, Federated three-layer MLP with Focal Loss and Federated Graph Convolutional Network (GCN) with Focal Loss), and two recent and advanced anomaly detection methods published in top AI Conferences, i.e., CARE~\parencite{zheng2025cluster}, ARC~\parencite{liu2024arc}. Notice that both federated MLP and federated GNN are trained using a federated averaging protocol while ARC is a generalist graph anomaly detection model trained on all sub-datasets. The rest baseline models are trained independently on each sub-dataset (i.e., each country). Figure \ref{fig:detection_model_results} reports the performance of our method relative to these baselines. 

To assess the benefits of federated learning, we compare the standard MLP with the Federated MLP. As shown in Figure \ref{fig:detection_model_results}, the Federated MLP with focal loss achieves higher AUCROC and AUPRC, while also reducing both Type I and Type II errors. We attribute these performance gains to the shared global model parameters maintained under the federated paradigm. By aggregating model updates rather than raw data, federated learning enables the model to capture cross-country behavioral patterns that would be inaccessible in isolated training. This promotes better generalization, reduces overfitting to local data, and allows each client model to benefit from the implicit knowledge of all others without compromising data privacy. In effect, the federated setting provides the model with a richer representation space, boosting predictive stability and robustness across markets.

To understand the value of incorporating structural information from graph data, we further compare Federated GCN with Federated MLP. The results show that Federated GCN achieves higher AUPRC, lower Type I error, and competitive AUCROC, while incurring only a modest increase in Type II error. The performance advantage of Federated GCN stems from its ability to leverage relational structures that encode how accounts interact with one another. These graph-based dependencies capture behavioral patterns that are difficult to learn from tabular features alone. By integrating both local node attributes and cross-account connectivity, Federated GCN can more effectively distinguish normal account behavior from suspicious activities. This structural signal proves pivotal in further boosting detection accuracy.

Two advanced graph anomaly detection methods (i.e., CARE, ARC) achieves better performance than Federated GNN with respect to AUPRC score and lower Type II error due to their special design for identifying the anomalies against the normal nodes, such as maximizing the difference of representations between the anomalies and normal nodes. Overall, our method achieves stronger performance than most competing approaches in terms of AUCROC and AUPRC. Although Random Forest achieves a higher AUPRC than our model, it does so at the expense of substantially higher Type II error (54.5\%), compared with only 5.53\% for our method. This trade-off reflects a model that favors positive predictions but sacrifices recall on the negative class, making it less suitable for money-laundering detection.

\subsubsection{Transferability of Detection Module}

\begin{table}[t]
    \centering
    \caption{Cross-market generalization under leave-one-country-out evaluation}
    \resizebox{\textwidth}{!}{
    \begin{tabular}{l|cccc|cccc}
        \hline
        \multirow{2}{*}{\textbf{Market}} 
        & \multicolumn{4}{c|}{\textbf{All Countries}} 
        & \multicolumn{4}{c}{\textbf{Leave Germany out}} \\
        \cline{2-9}
         & AUCROC & AUPRC & Type I & Type II
         & AUCROC & AUPRC & Type I & Type II \\
        \hline
        United States  & 0.9830 & 0.7135 & 0.0513 & 0.0534 & 0.9844 & 0.7257 & 0.0483 & 0.0575 \\
        Germany        & 0.9791 & 0.6047 & 0.0720 & 0.0421 & 0.9792 & 0.5970 & 0.0682 & 0.0512 \\
        France         & 0.9792 & 0.5638 & 0.0704 & 0.0387 & 0.9801 & 0.5674 & 0.0676 & 0.0464 \\
        Italy          & 0.9784 & 0.5678 & 0.0670 & 0.0571 & 0.9792 & 0.5724 & 0.0641 & 0.0645 \\
        Spain          & 0.9809 & 0.5246 & 0.0657 & 0.0334 & 0.9814 & 0.5222 & 0.0629 & 0.0373 \\
        China          & 0.9759 & 0.5131 & 0.0713 & 0.0325 & 0.9755 & 0.5096 & 0.0656 & 0.0423 \\
        Rest Countries & 0.9865 & 0.7791 & 0.0416 & 0.0415 & 0.9864 & 0.7735 & 0.0398 & 0.0558 \\
        \hline
        Overall        & 0.9804 & 0.6095 & 0.0628 &\textbf{0.0427} & 0.9809 & 0.6097 & 0.0595 & 0.0507 \\
        \hline\hline

        \multirow{2}{*}{\textbf{Market}} 
        & \multicolumn{4}{c|}{\textbf{Leave United States out}} 
        & \multicolumn{4}{c}{\textbf{Leave China out}} \\
        \cline{2-9}
         & AUCROC & AUPRC & Type I & Type II
         & AUCROC & AUPRC & Type I & Type II \\
        \hline
        United States  & 0.9787 & 0.6523 & 0.0498 & 0.0927 & 0.9821 & 0.7139 & 0.0485 & 0.0796 \\
        Germany        & 0.9799 & 0.6093 & 0.0660 & 0.0398 & 0.9781 & 0.5869 & 0.0671 & 0.0603 \\
        France         & 0.9808 & 0.5976 & 0.0655 & 0.0588 & 0.9781 & 0.5702 & 0.0672 & 0.0604 \\
        Italy          & 0.9819 & 0.6199 & 0.0599 & 0.0626 & 0.9783 & 0.5809 & 0.0635 & 0.0645 \\
        Spain          & 0.9818 & 0.5529 & 0.0599 & 0.0452 & 0.9805 & 0.5230 & 0.0625 & 0.0373 \\
        China          & 0.9764 & 0.5266 & 0.0643 & 0.0488 & 0.9733 & 0.4951 & 0.0635 & 0.0520 \\
        Rest Countries & 0.9873 & 0.8001 & 0.0402 & 0.0399 & 0.9866 & 0.7829 & 0.0389 & 0.0590 \\
        \hline
        Overall        & \textbf{0.9810} & \textbf{0.6227} & \textbf{0.0579} & 0.0554 & 0.9796 & 0.6075 & 0.0588 & 0.0590 \\
        \hline
    \end{tabular}}
    \label{tab:leave_one_country_out}
\end{table}

Although the proposed system demonstrates strong in-sample performance, a central question is whether the learned financial intelligence generalizes to markets unseen during training. To examine this, we conduct a leave-one-country-out evaluation, in which the model is trained on all countries except one and then tested on the excluded market, mimicking deployment in a new or emerging jurisdiction. As shown in Table~\ref{tab:leave_one_country_out}, overall detection performance remains remarkably stable across all hold-out scenarios. Relative to the full-data benchmark (overall AUCROC = 0.9804, AUPRC = 0.6095), excluding Germany has virtually no impact on global performance (AUCROC = 0.9809), while excluding the United States slightly improves overall AUPRC to 0.6227, indicating that the learned representations are not overly dependent on any single market. At the market (country) level, performance degradation on the held-out country is modest and economically interpretable. When the United States is excluded from training, the U.S. subgraph exhibits an increase in Type II error from 5.3\% to 9.3\%, reflecting the absence of high-volume, highly connected transaction patterns during training. At the same time, detection performance in other countries improves, with AUPRC increasing to 0.6093 in Germany and 0.6199 in Italy. This pattern is consistent with the disproportionate scale of the U.S. market (approximately 855k transactions) relative to other countries (each below 250k), suggesting that training without the dominant market encourages the model to learn more balanced and transferable transaction structures. A similar but weaker effect is observed when China is held out. The China subgraph experiences only a moderate increase in Type II error (from 3.3\% to 5.2\%), while overall performance remains close to the baseline (AUCROC = 0.9796), indicating robust generalization even to structurally distinct financial ecosystems. Across all settings, detection performance in smaller and less-connected markets (“Rest Countries”) remains consistently high, with AUCROC exceeding 0.986. 

Overall, these results provide strong evidence that the proposed system captures structural laundering patterns that generalize across countries rather than country-specific transaction idiosyncrasies. While excluding a dominant financial hub predictably reduces performance in that specific market, global detection capability remains stable and, in some cases, improves elsewhere. This highlights the system’s potential to support scalable cross-border deployment and to bootstrap effective anti-money-laundering detection in new or data-scarce markets without extensive local historical data.

\begin{table}[t]
    \centering
    \caption{Communication frequency vs. performance in federated learning}
    \label{tab:communication_fed}
    \resizebox{\textwidth}{!}{
    \begin{tabular}{l|cccc|cccc}
        \hline
        \multirow{2}{*}{\textbf{Market}} 
        & \multicolumn{4}{c|}{\textbf{frequency=1 }} 
        & \multicolumn{4}{c}{\textbf{frequency=1/2 }} \\
        \cline{2-9}
         & AUCROC & AUPRC & Type I & Type II
         & AUCROC & AUPRC & Type I & Type II \\
        \hline
        United States  & 0.9790 & 0.6659 & 0.0230 & 0.1935 & 0.9795 & 0.6812 & 0.0392 & 0.1048 \\
        Germany        & 0.9766 & 0.5584 & 0.0296 & 0.2059 & 0.9787 & 0.5884 & 0.0524 & 0.1081 \\
        France         & 0.9773 & 0.5301 & 0.0298 & 0.1734 & 0.9785 & 0.5391 & 0.0545 & 0.0805 \\
        Italy          & 0.9787 & 0.5691 & 0.0253 & 0.2284 & 0.9790 & 0.5684 & 0.0484 & 0.1105 \\
        Spain          & 0.9784 & 0.5073 & 0.0254 & 0.2083 & 0.9794 & 0.5036 & 0.0504 & 0.0884 \\
        China          & 0.9726 & 0.4735 & 0.0304 & 0.2520 & 0.9727 & 0.4968 & 0.0526 & 0.1122 \\
        Rest Countries & 0.9858 & 0.7883 & 0.0229 & 0.1866 & 0.9864 & 0.7765 & 0.0328 & 0.0909 \\
        \hline
        Overall        & 0.9783 & 0.5847 & 0.0266 & 0.2069 & 0.9792 & 0.5934 & 0.0472 & 0.0993 \\
        \hline
    \end{tabular}}

    \resizebox{\textwidth}{!}{
    \begin{tabular}{l|cccc|cccc}
        \hline
        \multirow{2}{*}{\textbf{Market}} 
        & \multicolumn{4}{c|}{\textbf{frequency=1/4}} 
        & \multicolumn{4}{c}{\textbf{frequency=1/8}} \\
        \cline{2-9}
         & AUCROC & AUPRC & Type I & Type II
         & AUCROC & AUPRC & Type I & Type II \\
        \hline
        United States  & 0.9830 & 0.7135 & 0.0513 & 0.0534 & 0.9824 & 0.7195 & 0.0494 & 0.0635 \\
        Germany        & 0.9791 & 0.6047 & 0.0720 & 0.0421 & 0.9792 & 0.5944 & 0.0672 & 0.0501 \\
        France         & 0.9792 & 0.5638 & 0.0704 & 0.0387 & 0.9807 & 0.5849 & 0.0670 & 0.0480 \\
        Italy          & 0.9784 & 0.5678 & 0.0670 & 0.0571 & 0.9792 & 0.5899 & 0.0635 & 0.0700 \\
        Spain          & 0.9809 & 0.5246 & 0.0657 & 0.0334 & 0.9809 & 0.5277 & 0.0625 & 0.0452 \\
        China          & 0.9759 & 0.5131 & 0.0713 & 0.0325 & 0.9757 & 0.5151 & 0.0657 & 0.0488 \\
        Rest Countries & 0.9865 & 0.7791 & 0.0416 & 0.0415 & 0.9873 & 0.7947 & 0.0394 & 0.0526 \\
        \hline
        Overall        & 0.9804 & 0.6095 & 0.0628 & 0.0427 & \textbf{0.9808} & \textbf{0.6180} & 0.0593 & 0.0540 \\
        \hline
    \end{tabular}}
\end{table}

\subsubsection{Communication Frequency of Detection Module}
To quantify the efficiency–accuracy trade-off in federated learning, we vary the communication frequency while holding the total training budget fixed. We fixed the training epochs to 4000 and vary the communication frequency from 1 to 1/2, 1/4, and 1/8 (communicate once every 8 epochs). Notice that the lower communication frequency mean the lower local updates before aggregation. Table~\ref{tab:communication_fed} reports a clear and economically meaningful pattern. When communication occurs every epoch (frequency=1), the system exhibits severe underperformance, with an overall AUPRC of 0.5847 and a Type II error rate of 20.7\%, indicating that overly frequent synchronization prevents institutions from adequately learning localized transaction structures. Halving the communication interval (frequency=1/2) improves AUPRC by nearly 0.9\% (from 0.5847 to 0.5934) and cuts the Type II error rate by more than half (from 20.7\% to 9.9\%), suggesting substantial gains from allowing richer local updates before aggregation. Performance improves when we further halve the communication frequencies. At frequency=1/4, the system achieves an overall AUPRC of 0.6095 with a 2.5\% improvement relative to frequency=1, while greatly reducing the Type II error from 20.69\% to 4.3\%. This improvement is consistent across all major countries/markets for example, in the United States subgraph, the missed-detection rate drops from 19.4\% at frequency=1 to 5.3\% at frequency=1/4, while AUPRC increases from 0.666 to 0.714. Importantly, these gains are obtained with a fourfold reduction in communication rounds, highlighting a strong efficiency payoff. Moving to an even lower communication regime (frequency=1/8) yields diminishing but still positive returns. While overall AUPRC increases slightly to 0.6180, Type II error rises modestly to 5.4\%, indicating a trade-off between recall and communication sparsity. This pattern suggests that excessively infrequent synchronization can introduce model drift across institutions, particularly in heterogeneous markets such as China and Italy. Nevertheless, even at frequency=1/8, performance remains substantially superior to high-frequency communication, with a 74\% reduction in missed detections relative to frequency=1.

These results demonstrate that federated efficiency in cross-border anti-money-laundering systems is non-monotonic in communication frequency. Moderate synchronization intervals dominate both extremes by simultaneously improving detection accuracy and reducing communication costs. This implies that regulators and financial institutions can achieve higher detection effectiveness while substantially lowering coordination and compliance overhead by avoiding overly frequent model aggregation. The findings provide practical guidance for designing scalable, privacy-preserving financial intelligence systems in globally distributed settings.

\subsubsection{Granularity of super-node sharing}
\label{supernode_high_order_results}1
In this section, we aim to investigate the impact of varying the granularity of supernode feature sharing by comparing two aggregation strategies: using only the mean of connected node features (first-order statistical moment) versus using both the mean and variance (first-order and second-order statistical moment). Table~\ref{tab:supernode_aggregation} shows that the differences in overall performance between the two approaches are minimal. While the mean + variance aggregation slightly reduces Type II errors in some markets (e.g., Germany, Spain, China), it also increases Type I errors in others (e.g., United States, France). Across all markets, AUCROC and AUPRC remain largely stable, indicating that incorporating higher-order moments beyond the first has diminishing returns. These results suggest that capturing the first moment (mean) of supernode features is sufficient for effective detection, and additional statistics such as variance provide only marginal benefits relative to the added complexity of computation and communication in a cross-bank setting. Therefore, a simple mean-based aggregation offers a favorable trade-off between performance and efficiency for supernode sharing in heterogeneous financial networks.
\begin{table}[htbp]
\centering
\caption{Comparison of supernode feature aggregation strategies: Mean aggregation vs. Mean + Variance aggregation. Performance is reported for each market in terms of AUC, Average Precision (AP), Type I error, and Type II error.}
\resizebox{\textwidth}{!}{
\begin{tabular}{lcccc|cccc}
\hline
\multirow{2}{*}{\textbf{Market}} & \multicolumn{4}{c|}{\textbf{Mean Aggregation}} & \multicolumn{4}{c}{\textbf{Mean + Variance Aggregation}} \\
\cline{2-9}
 & \textbf{AUCROC} & \textbf{AUPRC} & \textbf{Type I} & \textbf{Type II} & \textbf{AUCROC} & \textbf{AUPRC} & \textbf{Type I} & \textbf{Type II} \\
\hline
United States & 0.9824 & 0.7195 & 0.0494 & 0.0635 & 0.9825 & 0.6951 & 0.0540 & 0.0625 \\
Germany       & 0.9792 & 0.5944 & 0.0672 & 0.0501 & 0.9789 & 0.5960 & 0.0719 & 0.0375 \\
France        & 0.9807 & 0.5849 & 0.0670 & 0.0480 & 0.9802 & 0.5850 & 0.0717 & 0.0433 \\
Italy         & 0.9792 & 0.5899 & 0.0635 & 0.0700 & 0.9788 & 0.5921 & 0.0658 & 0.0497 \\
Spain         & 0.9809 & 0.5277 & 0.0625 & 0.0452 & 0.9807 & 0.5372 & 0.0669 & 0.0334 \\
China         & 0.9757 & 0.5151 & 0.0657 & 0.0488 & 0.9751 & 0.5064 & 0.0697 & 0.0325 \\
Rest Countries & 0.9873 & 0.7947 & 0.0394 & 0.0526 & 0.9861 & 0.8003 & 0.0412 & 0.0351 \\
\hline
\textbf{Overall} & \textbf{0.9808} & \textbf{0.6180} & \textbf{0.0593} & 0.0540 & 0.9803 & 0.6160 & 0.0630 & \textbf{0.0420} \\
\hline
\end{tabular}}
\label{tab:supernode_aggregation}
\end{table}

\begin{table}[h!]
\centering
\caption{Statistics of transaction by country}
\resizebox{0.9\textwidth}{!}{
\begin{tabular}{lcccc}
\hline
\textbf{Country} & \textbf{\#Accounts} & \textbf{\#Total Transactions} & \textbf{\#Illicit Transactions} & \textbf{Ratio of Illicity}\\
\hline
United States   & 396,370 & 4,175,722   & 3,489     & 0.084\%\\
Germany         & 197,543 & 1,599,846   & 1,542     & 0.096\%\\
France          & 168,372 & 1,284,859   & 1,145     & 0.089\%\\
Italy           & 149,425 & 1,082,159   & 977       & 0.090\%\\
Spain           & 137,726 & 989,019     & 911       & 0.092\%\\
China           & 130,099 & 973,821     & 1,050     & 0.011\%\\
Rest Countries  & 278,480 & 1,700,615   & 1,284     & 0.076\%\\
\hline
Overall         & 208,288 & 1,686,577   & 1,485     & 0.077\%\\
\hline
\end{tabular}}
\label{tab:country_graph_stats_full}
\end{table}

\subsubsection{Scalability to Enterprise-Level Transaction Volumes}
\label{scalability_analysis}
We evaluate the detection module on the full IBM AML dataset comprising over 5 million transactions (Table \ref{tab:country_graph_stats_full}). This setting presents substantially greater challenges than the primary experiments: the illicit transaction ratio drops to approximately 0.08\%—nearly five times more imbalanced than the 1.4 million transaction subset—and the labeled training data is limited to 20\% of transactions, reflecting realistic constraints on investigator capacity for manual labeling.

Table~\ref{scalability_analysis} reports the results. The baseline GNN with focal loss exhibits severe majority-class bias, achieving low Type I error (0.33\%) but unacceptably high Type II error (62.96\%), indicating that the model learns to classify nearly all transactions as legitimate. Our method, combining federated graph learning with structure-preserving downsampling (20:1 ratio), substantially improves detection sensitivity: Type II error decreases from 62.96\% to 23.13\%—a 63\% reduction in missed illicit transactions. This improvement comes at the cost of moderately higher Type I error (6.79\%), reflecting the fundamental precision-recall trade-off under extreme imbalance.

Several observations merit discussion. First, the performance gap between the 1.4M and 5M transaction settings highlights the compounding difficulty of extreme class imbalance: at 0.08\% positive rate, even focal loss struggles to prevent majority-class dominance. Second, AUPRC scores (0.33 overall) are lower than in the smaller dataset, consistent with the information-theoretic limits of learning from very few positive examples per training batch. Third, cross-country heterogeneity persists, with "Rest Countries" showing the weakest performance (Type II = 29.02\%), potentially reflecting greater distributional shift in smaller, less-regulated markets.

These results suggest two directions for improvement that we leave to future work: (1) increasing the proportion of labeled training data, which is constrained by investigator capacity in practice, and (2) developing curriculum learning strategies that progressively increase task difficulty as the model learns basic laundering patterns. Despite these limitations, the substantial reduction in Type II error demonstrates that our framework provides meaningful detection improvements even at enterprise scale, where baseline methods effectively fail.

\begin{table*}[ht!]
\centering
\caption{Performance Comparison Between Our Method and the Baseline Method (GNN with Focal Loss)}
\small
\resizebox{0.95\textwidth}{!}{
\begin{tabular}{l|cccc|cccc}
\hline
\multirow{2}{*}{\textbf{Markets}} 
& \multicolumn{4}{c|}{\textbf{Our Method}} 
& \multicolumn{4}{c}{\textbf{GNN + Focal Loss}} \\
\cline{2-9}
& AUCROC & AUPRC & Type I & Type II 
& AUCROC & AUPRC & Type I & Type II \\
\hline
United States & 0.9667 & 0.4407 & 0.0554 & 0.2164 & 0.9693 & 0.3867 & 0.0036 & 0.5550 \\
Germany       & 0.9619 & 0.3941 & 0.0688 & 0.2156 & 0.9684 & 0.3728 & 0.0034 & 0.5932 \\
France        & 0.9608 & 0.3438 & 0.0694 & 0.2074 & 0.9664 & 0.3593 & 0.0033 & 0.5841 \\
Italy         & 0.9510 & 0.2937 & 0.0776 & 0.2570 & 0.9635 & 0.2655 & 0.0027 & 0.7110 \\
Spain         & 0.9556 & 0.2887 & 0.0804 & 0.2373 & 0.9644 & 0.3565 & 0.0028 & 0.6296 \\
China         & 0.9657 & 0.3129 & 0.0649 & 0.1952 & 0.9672 & 0.3285 & 0.0043 & 0.6155 \\
Rest Countries     & 0.9523 & 0.2231 & 0.0586 & 0.2902 & 0.9549 & 0.1842 & 0.0030 & 0.7186 \\
\hline
Overall             & 0.9591 & \textbf{0.3281} & 0.0679 & \textbf{0.2313} 
                    & \textbf{0.9649} & 0.3219 & \textbf{0.0033} & 0.6296 \\
\hline
\end{tabular}
}
\label{scalability_analysis}
\end{table*}

\subsection{Empirical Results of Cross-bank PPR}
\subsubsection{Interpretation via Identified Group Patterns}
\label{PPR_interpretation}

\begin{figure}[htbp]
    \centering
    \vspace{-3mm}
    \begin{subfigure}{0.5\linewidth}
        \centering
        \includegraphics[width=\linewidth]{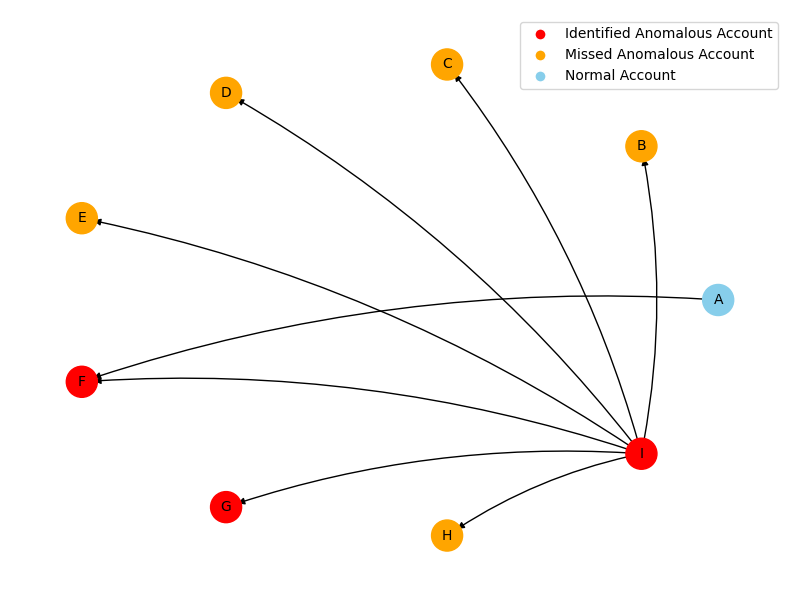}
        \caption{Pattern with Fan-out Structure}
    \end{subfigure}%
    \hfill
    \begin{subfigure}{0.5\linewidth}
        \centering
        \includegraphics[width=\linewidth]{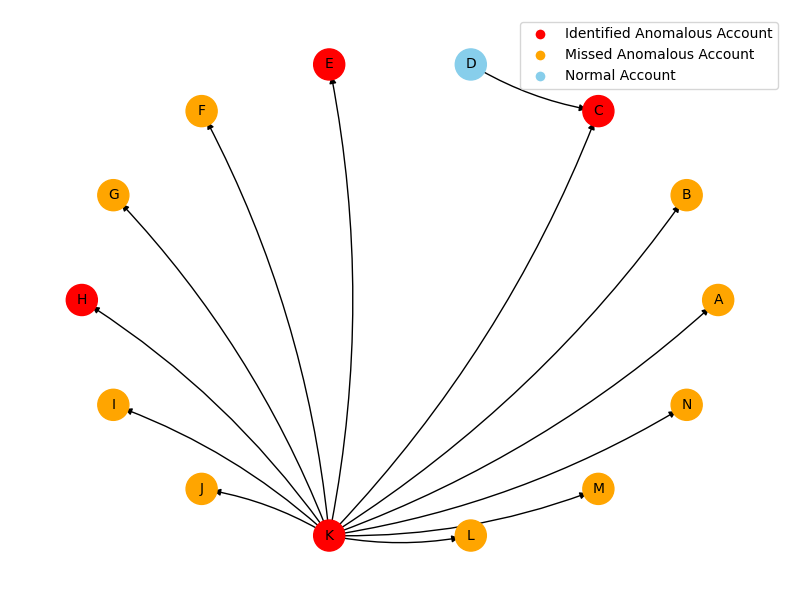}
        \caption{Pattern with Another Fan-out Structure}
    \end{subfigure}%
    \vspace{0.5em} 
    
    \vspace{-3mm}
    \begin{subfigure}{0.5\linewidth}
        \centering
        \includegraphics[width=\linewidth]{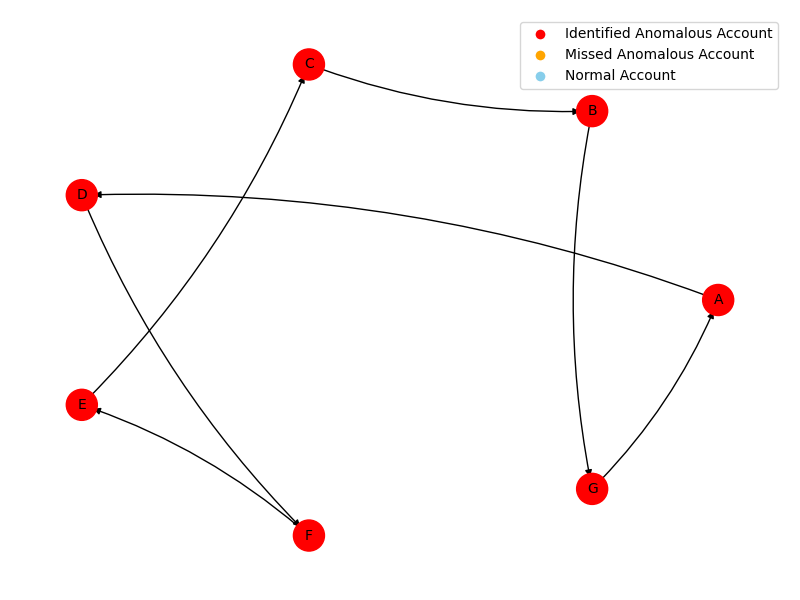}
        \caption{Pattern with Loop Structure}
    \end{subfigure}%
    \hfill
    \begin{subfigure}{0.5\linewidth}
        \centering
        \includegraphics[width=\linewidth]{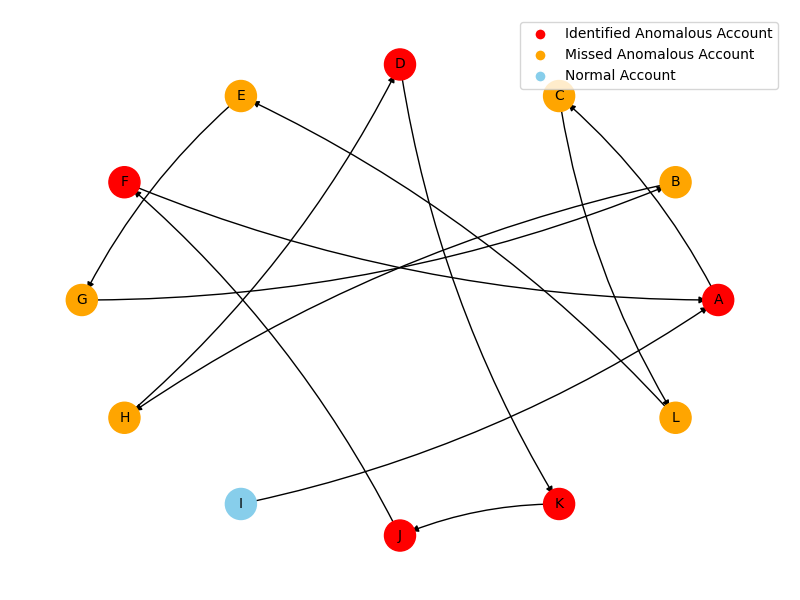}
        \caption{Another Pattern with Loop Structure}
    \end{subfigure}
    \vspace{0.5em} 
    
    \vspace{-3mm}
    \begin{subfigure}{0.5\linewidth}
        \centering
        \includegraphics[width=\linewidth]{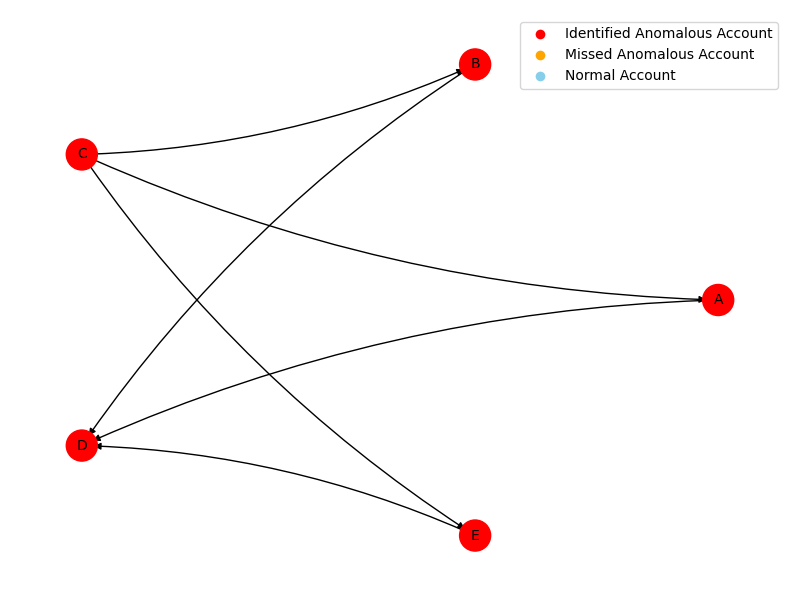}
        \caption{Pattern with Gather-scatter Structure}
    \end{subfigure}%
    \hfill
    \begin{subfigure}{0.5\linewidth}
        \centering
        \includegraphics[width=\linewidth]{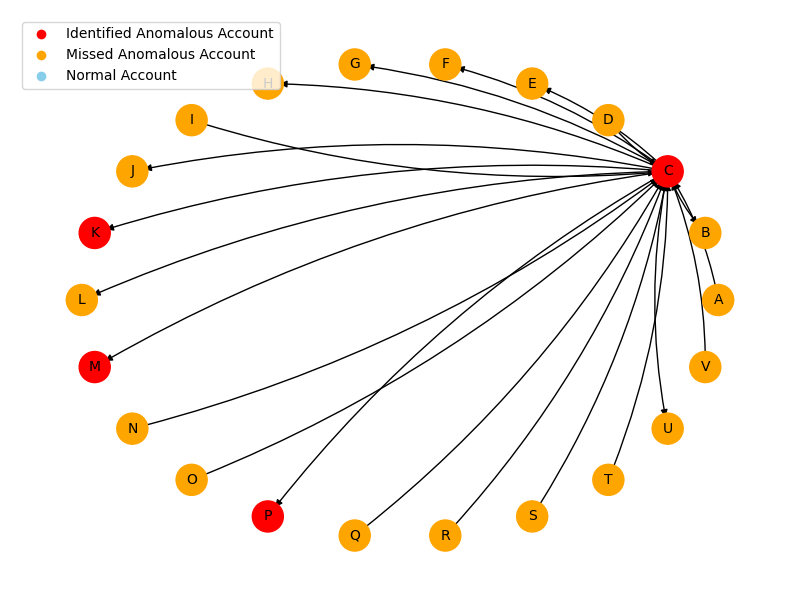}
        \caption{Pattern with Gather and Fan-out Structure}
    \end{subfigure}
    \caption{Identified Group Money-Laundering Patterns}
    \label{fig:money_laundering_demo}
\end{figure}

We examine the interpretability of group money-laundering patterns identified by the cross-bank Personalized PageRank method. Figure~\ref{fig:money_laundering_demo} illustrates representative patterns detected from the merged transaction clusters. We focus on four distinct types of money-laundering patterns: fan-out, loop, gather-scatter, and hybrid structures combining gather and fan-out mechanisms.

Figures~\ref{fig:money_laundering_demo} (a) and (b) depict fan-out structures, in which funds from a single account are distributed to multiple downstream accounts, reflecting the dispersal behavior commonly used to obscure illicit money flows. In Figure~\ref{fig:money_laundering_demo} (a), funds originate from account I and are distributed to accounts B, C, D, E, F, G, and H. Account A is a normal account incorrectly flagged as suspicious, highlighting the need for further expert review to refine labels. Importantly, our method correctly identifies the primary control account I, which orchestrates the laundering activity. Some accounts, such as B, C, D, E and H, are missed, which we attribute to the “bank isolation” issue, where complete group structures are fragmented across institutions, complicating detection. Similarly, in Figure~\ref{fig:money_laundering_demo} (b), our method correctly identifies account K as the principal node managing the funds.

Figures~\ref{fig:money_laundering_demo} (c) and (d) illustrate loop structures, where funds circulate among accounts before returning to the origin, representing more sophisticated laundering strategies. In Figure~\ref{fig:money_laundering_demo} (c), our method successfully identifies the full loop: funds flow from account A to D, then through accounts D–G, before returning to A. In Figure~\ref{fig:money_laundering_demo} (d), laundering begins at account B in China, with funds moving across accounts in Europe and the United States along multiple chains, including $B \rightarrow H \rightarrow D$ (Europe-based chain), $D \rightarrow K \rightarrow J \rightarrow F \rightarrow A$ (United States-based chain) and  $A \rightarrow C \rightarrow L$ (Europe-based chain). Our method captures the U.S.-based chain but misses the cross-border components, illustrating the challenges posed by fragmented international transactions. Despite partial detection, identifying even a segment of the loop provides valuable signals of illicit activity.

Figures~\ref{fig:money_laundering_demo} (e) and (f) demonstrate gather-scatter mechanisms and hybrid structure combining gather-scatter and fan-out mechanisms, where multiple accounts funnel funds into a central account before redistribution. In Figure~\ref{fig:money_laundering_demo} (e), account C distributes funds to accounts A, B, and E, and account D gather all money flows, completing the laundering cycle, which our method successfully detects. In Figure~\ref{fig:money_laundering_demo} (f), a more complex behavior emerges: account C gathers funds from multiple sources and redistributes them to multiple recipients. While only four of the 22 suspicious accounts are identified, account C is correctly flagged as the pivotal node orchestrating the activity.

These results provide a clear interpretation of group money-laundering behaviors by revealing the underlying relational structures within transaction networks. By identifying patterns such as fan-out, loop, gather-scatter, and hybrid mechanisms, the method highlights how funds flow between accounts and pinpoints the pivotal nodes that orchestrate these operations. Even when full laundering cycles are fragmented across banks or countries, the detection of key control accounts and partial transaction chains offers interpretable insights into the mechanisms of illicit activity. This structural perspective allows analysts to understand not only which accounts are suspicious, but also how money moves through the network, supporting more informed decision-making, targeted investigations, and effective regulatory interventions.

\begin{figure}[htbp]
    \centering
    \vspace{-3mm}
    \begin{subfigure}{0.5\linewidth}
        \centering
        \includegraphics[width=\linewidth]{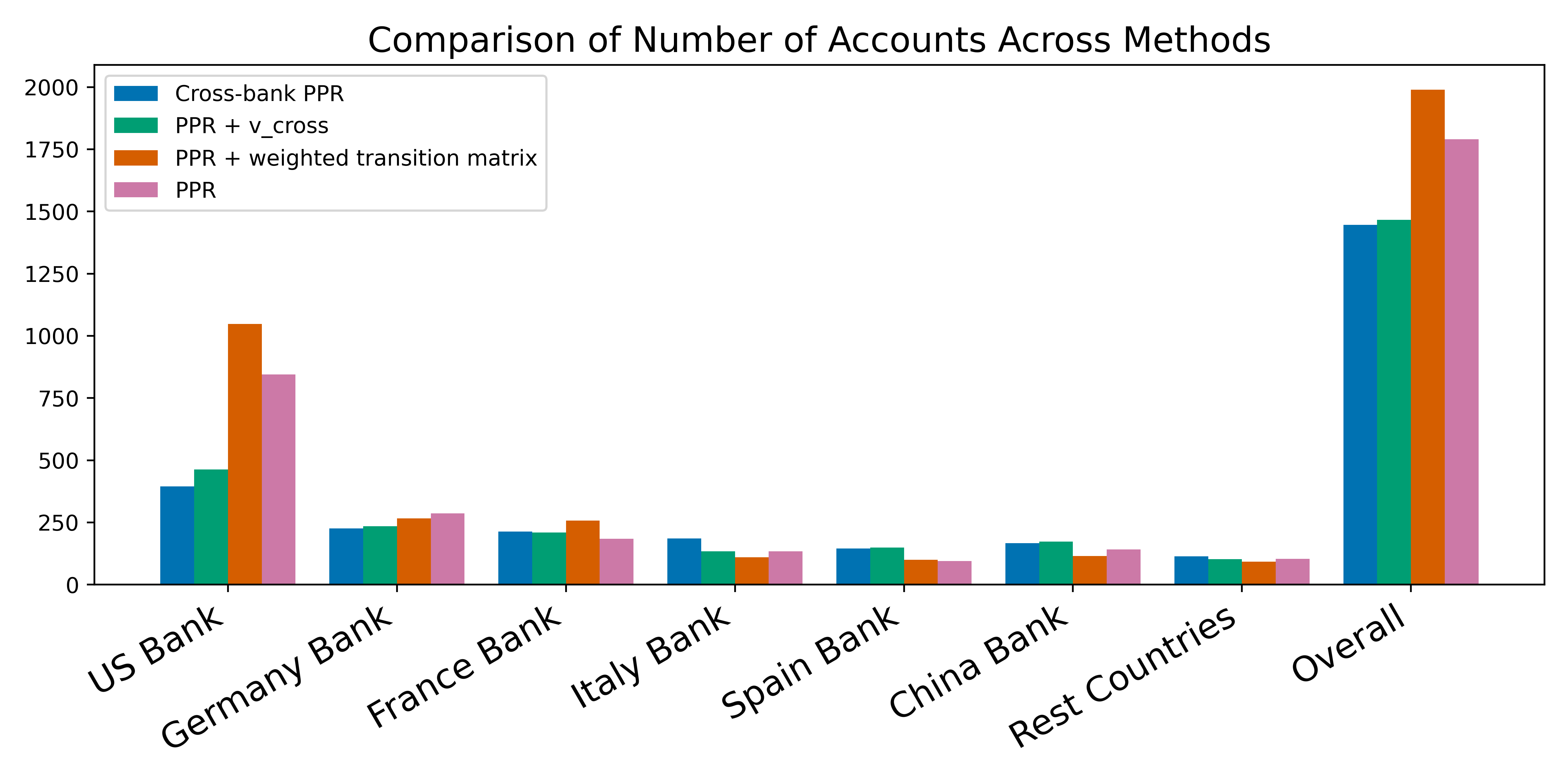}
        \begin{center}
            \caption{Total number of accounts included \\
            across all clusters}
        \end{center}
    \end{subfigure}%
    \hfill
    \begin{subfigure}{0.5\linewidth}
        \centering
        \includegraphics[width=\linewidth]{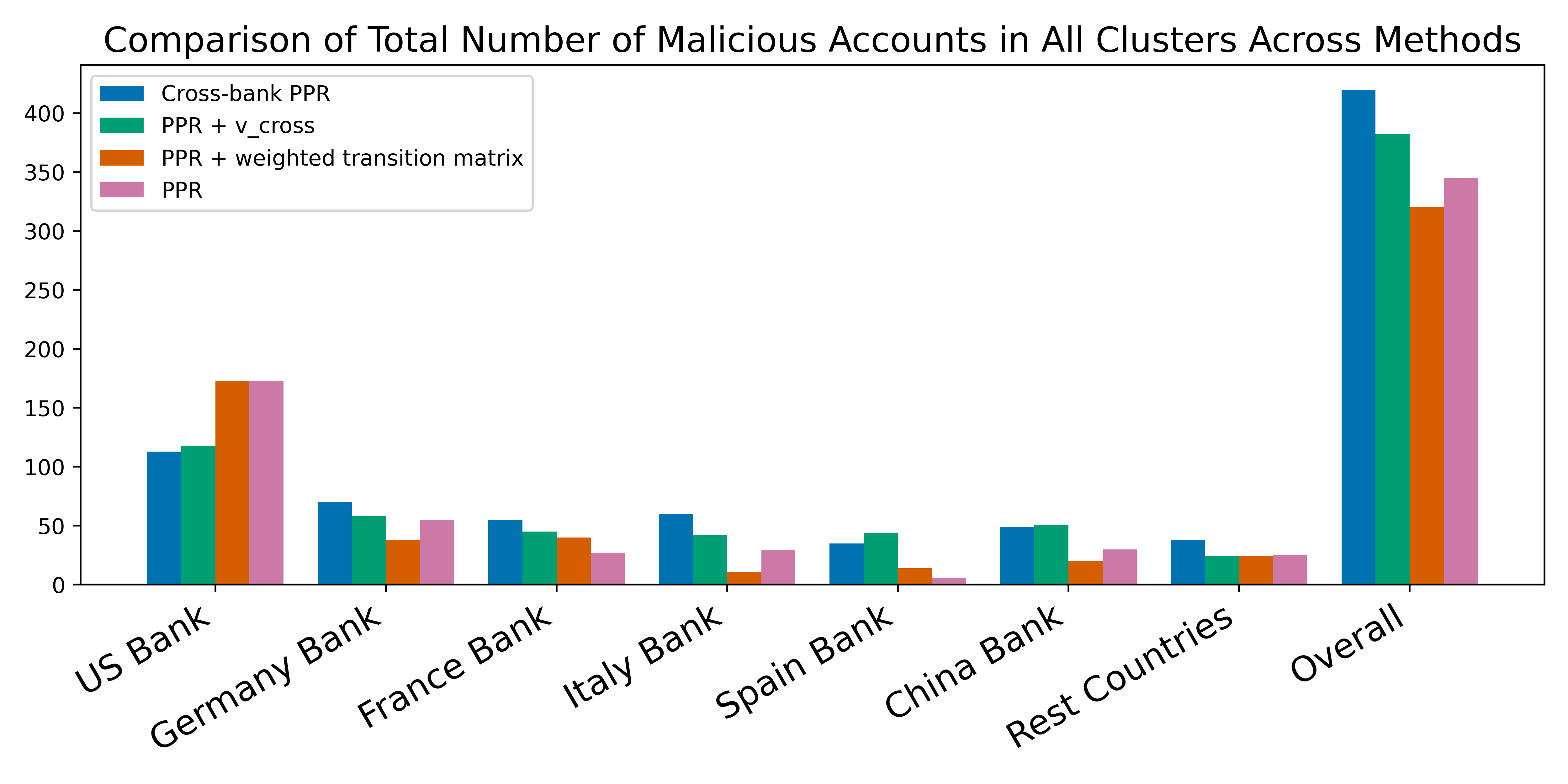}
        \begin{center}
            \caption{Number of malicious accounts  \\
            in all clusters}
        \end{center}
    \end{subfigure}%
    \vspace{0.5em} 
    
    \begin{subfigure}{0.5\linewidth}
        \centering
        \includegraphics[width=\linewidth]{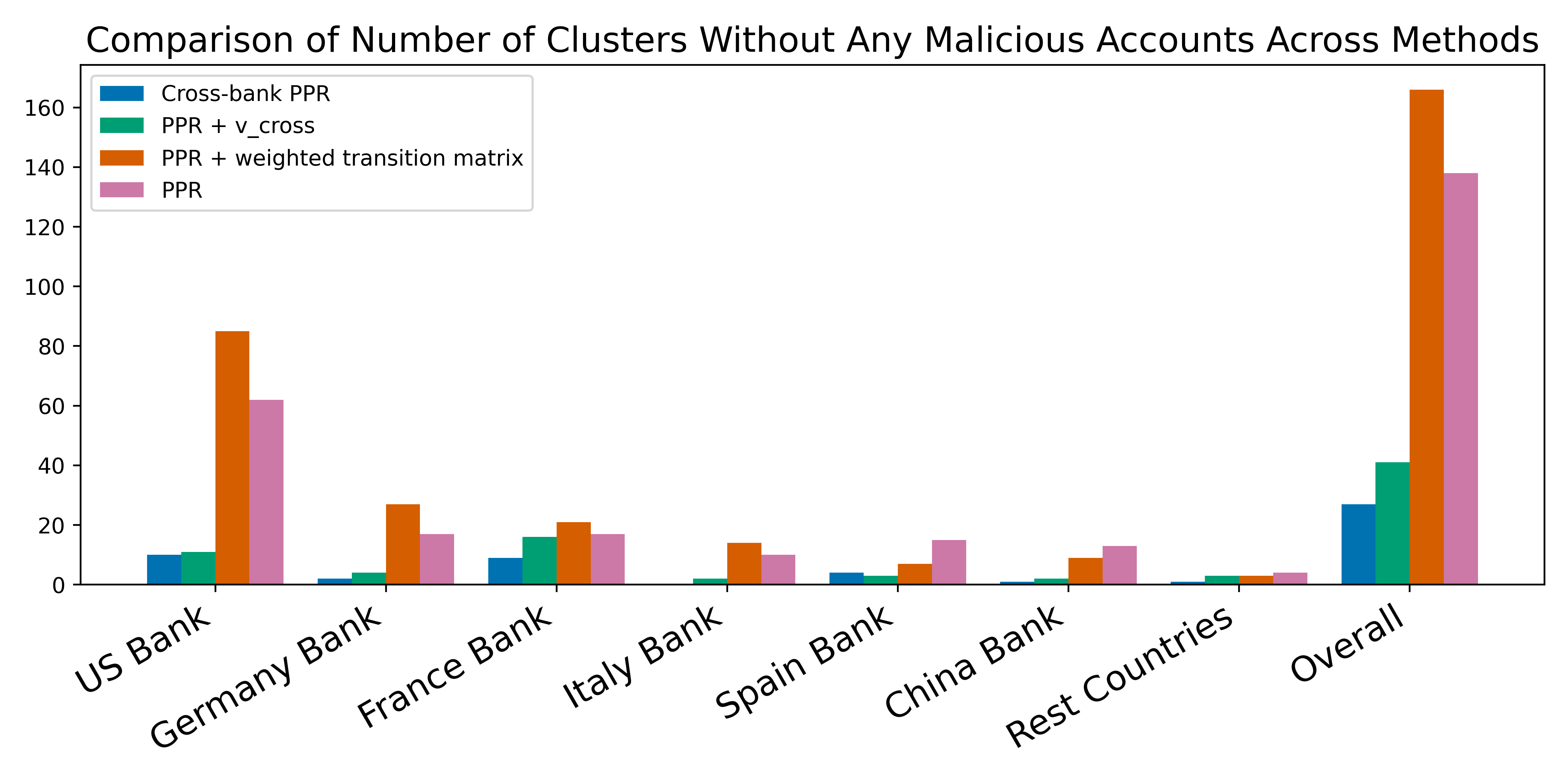}
        \begin{center}
            \caption{Number of clusters containing no \\
            malicious accounts}
        \end{center}
    \end{subfigure}%
    \hfill
    \begin{subfigure}{0.5\linewidth}
        \centering
        \includegraphics[width=\linewidth]{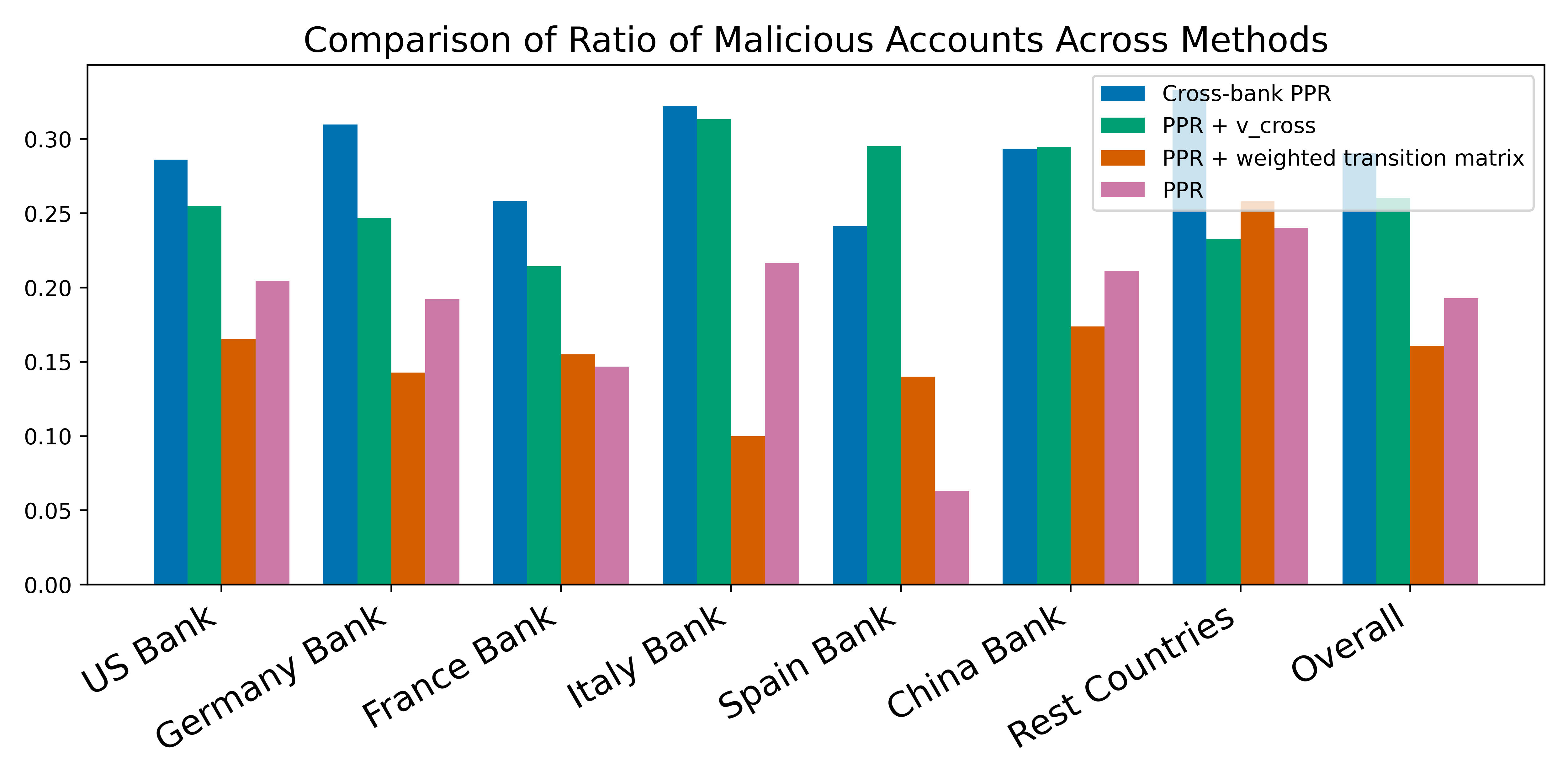}
        \begin{center}
            \caption{Overall proportion of malicious \\
            accounts}
        \end{center}
    \end{subfigure}
    \vspace{0.5em} 
    
    \caption{Comparison between Cross-bank PPR and PPR}
    \label{fig:PPR_evaluation}
\end{figure}

\subsubsection{Effectiveness of Cross-bank PPR}
In this subsection, we examine the hypothesis that Cross-bank PPR improves malicious-account detection and reduces false positives, compared with PPR. To verify this hypothesis, we compare cross-bank PPR with vanilla PPR. Figure \ref{fig:PPR_evaluation} reports summary statistics for clusters generated by Cross-bank PPR, its two ablated variants, and standard PPR. The statistics include: the total number of accounts included across all clusters, the number of identified malicious accounts, the number of clusters containing no malicious accounts, and the overall proportion of malicious accounts (malicious accounts divided by total accounts).
Combining the results from Figures \ref{fig:PPR_evaluation} (a) and (b), we observe that Cross-bank PPR identifies fewer suspicious accounts overall than the other methods, yet it discovers substantially more true malicious accounts. This pattern indicates that Cross-bank PPR is simultaneously more conservative in flagging suspicious neighbors but more accurate in finding genuinely malicious ones. This cautiousness is further supported by Figure \ref{fig:PPR_evaluation} (c), which shows that Cross-bank PPR produces the fewest clusters with no malicious accounts, indicating fewer false positives and more targeted clustering. Finally, Figure \ref{fig:PPR_evaluation} (d) demonstrates that Cross-bank PPR achieves the highest malicious-account identification rate across most countries, as well as the highest overall performance.

Next, we evaluate the effectiveness of two key components (i.e., introducing the cross-bank component $v_{cross}$ and adding the edge prediction score $\hat{y}_{ij}$ into the transition matrix) of cross-bank PPR compared with vanilla PPR. When the cross-bank suspicious signal ($v_{cross}$) is incorporated into PPR, we improve the capability of this new variant (`PPR + v\_cross' in Figure \ref{fig:PPR_evaluation}) detecting malicious accounts and reduce the false positives, compared with PPR. This can be verified by the fewer suspicious accounts overall than PPR, more true malicious accounts included across all clusters, and less clusters with no malicious accounts, shown in \ref{fig:PPR_evaluation} (a), (b) and (c). Different from the effectiveness of cross-bank suspicious signal ($v_{cross}$), when we only incorporate the edge malicious score ($\hat{y}_{ij}$) to update the weighted transition matrix into PPR (denoted as `PPR + weighted transition matrix' in Figure \ref{fig:PPR_evaluation}), we immediately observe the performance drops across all four statistics in Figure \ref{fig:PPR_evaluation}. 
Notice that Theorem~\ref{thm:detectability} reveals a detectability condition for identifying a money-laundering group using Personalized PageRank (PPR) and whether the money-laundering group is well-separated heavily rely on the quality of the weighted transition matrix. Therefore, we conjecture that the edge malicious score by a single bank is not reliable and incorporating the unreliable score into the weighted transition matrix only leads to the worse results than PPR. Crucially, when $\hat{y}_{ij}$ is combined with $v_{cross}$ in Cross-bank PPR, its previously unreliable signal is corrected, yielding superior performance. This demonstrates that the two components are complementary and jointly critical for achieving high precision and effectiveness in malicious-account detection.

Together, these results highlight that both components ($v_{cross}$ and $\hat{y}_{ij}$) play a critical role in improving the precision and effectiveness of Cross-bank PPR in identifying malicious financial activity. Unlike standard PPR, Cross-bank PPR identifies fewer accounts overall but captures substantially more true malicious accounts, demonstrating higher precision in detecting risky entities. This cautious yet effective behavior is a key advantage in financial monitoring.

Identifying module and decision-making module play complementary roles: Cross-bank PPR produces structured, interpretable group signals under fragmentation, while the hierarchical decision model translates these signals into economically optimal interventions. Importantly, the decision layer remains robust even when group structures are partially observed, reflecting realistic cross-border compliance conditions. 

\subsubsection{Effectiveness of Label Propagation}
\label{effectivenss_label_propagation}
\begin{table}[t]
\centering
\caption{Detection performance with and without label propagation.}
\label{tab:label_propagation}
\resizebox{\linewidth}{!}{
\begin{tabular}{l|cccc|cccc}
\hline
 & \multicolumn{4}{c|}{\textbf{Without Label Propagation}} 
 & \multicolumn{4}{c}{\textbf{With Label Propagation}} \\
\hline
\textbf{Market} 
 & AUCROC & AUPRC & Type I & Type II 
 & AUCROC & AUPRC & Type I & Type II \\
\hline
United States        & 0.9819 & 0.6878 & 0.0459 & 0.0786 & 0.9820 & 0.6903 & 0.0534 & 0.0655 \\
Germany   & 0.9790 & 0.6007 & 0.0702 & 0.0455 & 0.9790 & 0.6025 & 0.0695 & 0.0478 \\
France    & 0.9788 & 0.5595 & 0.0686 & 0.0557 & 0.9788 & 0.5622 & 0.0706 & 0.0526 \\
Italy     & 0.9782 & 0.5867 & 0.0634 & 0.0755 & 0.9784 & 0.5991 & 0.0650 & 0.0645 \\
Spain     & 0.9802 & 0.5356 & 0.0603 & 0.0530 & 0.9802 & 0.5360 & 0.0651 & 0.0491 \\
China     & 0.9742 & 0.4949 & 0.0601 & 0.0504 & 0.9745 & 0.5091 & 0.0677 & 0.0358 \\
Rest Countries & 0.9861 & 0.8004 & 0.0398 & 0.0526 & 0.9862 & 0.8038 & 0.0410 & 0.0415 \\
\hline
\textbf{Overall} 
               & 0.9798 & 0.6094 & 0.0583 & 0.0588 
               & \textbf{0.9799} & \textbf{0.6147} & 0.0618 & \textbf{0.0510} \\
\hline
\end{tabular}
}
\end{table}

In this section, we validate the hypothesis that the discovered group-wise laundering patterns can be further exploited to improve detection performance through label propagation, as summarized in Table~\ref{tab:label_propagation}. Overall, incorporating propagated group-level signals consistently improves detection accuracy across all evaluated transaction subgraphs. In particular, we observe systematic gains in AUPRC for every bank-level subgraph, with especially notable improvements in more challenging settings such as China Bank and Italy Bank. These results indicate that label propagation is effective at recovering suspicious nodes and transactions that initially receive low confidence scores but are structurally connected to high-confidence laundering groups. Across all subgraphs, the AUCROC remains stable or slightly improves after label propagation, indicating that the refinement step does not degrade the model’s global discrimination ability. While we observe modest increases in Type I error rates for some graphs (e.g., US Bank and Spain Bank), these increases are accompanied by notable reductions in Type II errors in several cases, most prominently for China Bank and the aggregated Rest Countries graph. This trade-off reflects the intended behavior of label propagation: prioritizing recall of coordinated laundering activity by propagating malicious signals within densely connected groups, even at the cost of a small increase in false positives. These results confirm that group-wise laundering patterns captured by cross-bank PPR provide complementary information beyond edge-level predictions. By integrating this information through label propagation, the detection modu moves beyond isolated transaction analysis and leverages collective behavioral evidence, leading to more robust identification of coordinated money laundering activity across heterogeneous financial networks. We next evaluate the economic value of label propagation by incorporating it into the decision-making module to quantify its impact on prevented losses and resource allocation in the follow section.

\subsection{Empirical Results of Decision-Making Module}
\subsubsection{Economic Value of Hierarchical Decision-Making Module}
A central question for financial institutions is the economic value generated by a more sophisticated intervention mechanism. Specifically, how much financial loss can the proposed hierarchical decision-making module avert relative to traditional, fixed-threshold approaches? To address this, we estimate both the total financial loss attributable to money-laundering transactions and the prevented loss achieved when intervention policies are deployed by (i) a decision-making module with a fixed threshold and (ii) our hierarchical decision-making module. Table~\ref{tab:economic_values} summarizes these results in Panels A and B, respectively.

Our empirical analysis is conducted on the IBM Anti–Money Laundering (AML) dataset, which contains more than 1.4 million transaction records collected between September 1 and September 10, 2022. To emulate a realistic operational environment, we partition the data chronologically: the first 80\% of transactions (September 1–8) are sued for model training, while the remaining 20\% (September 8–10) are used exclusively for evaluation. Within the first part, we further split the data into equal halves for training and validation. We define  \textit{Threshold} as the threshold to freeze the suspicious transactions, \textit{Total Loss} as the aggregated dollar amount associated with successfully executed illicit transactions and \textit{Prevented Loss} as the amount that would have been blocked had the recommended intervention (i.e., transaction freezing) been deployed at the time of execution. We first train the federated graph-based detection modu (Section~\ref{sec:detection_model}) on the training data to produce transaction-level maliciousness scores across the training, validation and test sets. These scores in training set are then used to train the hierarchical decision-making module, which learns optimal thresholding policies for interventions. Hyperparameters are selected using the validation set. During evaluation, we also track Type I and Type II errors to capture the operational quality of intervention decisions. A missed freeze on a money-laundering transaction is counted as a false negative (Type II error), while an unnecessary freeze on a legitimate transaction is counted as a false positive (Type I error). These error metrics reflect the business risks of either failing to prevent financial crime or unnecessarily disrupting customer activity.

\begin{table}[htp]
\centering
\caption{Comparison of decision-making module with Fixed Threshold and Hierarchical decision-making module Performance Across Countries. The fixed threshold in Panel A is selected based on the overall threshold in Panel B for fair comparison.}
\resizebox{0.9\textwidth}{!}{
\begin{tabular}{lcccccc}
\hline
\multicolumn{7}{c}{\textbf{Panel A: Decision-making module with fixed threshold}} \\
\hline
\textbf{Market} & \textbf{Threshold} & \textbf{Total Loss} & \textbf{Prevented Loss} & \textbf{Prevented Loss Ratio} & \textbf{Type I Error} & \textbf{Type II Error} \\
\hline
United States   & 0.3974 & \$456.41 Million & \$252.43 Million & 0.5531 & 0.0178 & 0.4152 \\
Germany         & 0.3974 & \$357.24 Million & \$201.59 Million & 0.5643 & 0.0125 & 0.5830 \\
France          & 0.3974 & \$212.05 Million & \$110.13 Million & 0.5194 & 0.0119 & 0.5546 \\
Italy           & 0.3974 & \$223.72 Million & \$101.80 Million & 0.4550 & 0.0089 & 0.6141 \\
Spain           & 0.3974 & \$280.05 Million & \$57.51 Million  & 0.2053 & 0.0094 & 0.5933 \\
China           & 0.3974 & \$292.43 Million & \$147.15 Million & 0.5032 & 0.0116 & 0.6140 \\
Rest Countries  & 0.3974 & \$118.45 Million & \$88.10 Million  & 0.7438 & 0.0069 & 0.4610 \\
\textbf{Overall} & 0.3974 & \$1940.35 Million & \$958.70 Million & 0.4941 & 0.0113 & 0.5479 \\
\hline
\end{tabular}}
\resizebox{0.9\textwidth}{!}{
\begin{tabular}{lcccccc}
\multicolumn{7}{c}{\textbf{Panel B: Hierarchical Decision-Making Module (without Label Propagation)}} \\
\hline
\textbf{Market} & \textbf{Threshold} & \textbf{Total Loss} & \textbf{Prevented Loss} & \textbf{Prevented Loss Ratio} & \textbf{Type I Error} & \textbf{Type II Error} \\
\hline
United States & 0.3163 & \$456.41 Million & \$455.29 Million & 0.9975 & 0.1029 & 0.0554 \\
Germany & 0.4195 & \$357.24 Million & \$343.75 Million & 0.9622 & 0.0348 & 0.3138 \\
France & 0.5226 & \$212.05 Million & \$24.32 Million & 0.1147 & 0.0083 & 0.6753 \\
Italy & 0.3163 & \$223.72 Million & \$223.60 Million & 0.9995 & 0.1023 & 0.0403 \\
Spain & 0.4711 & \$280.05 Million & \$80.30 Million & 0.2868 & 0.0150 & 0.5075 \\
China & 0.3679 & \$292.43 Million & \$292.01 Million & 0.9985 & 0.0847 & 0.1462 \\
Rest Countries & 0.3679 & \$118.45 Million & \$118.36 Million & 0.9993 & 0.0654 & 0.0487 \\
\textbf{Overall} & 0.3974 & \$1940.35 Million & \$1537.64 Million & 0.7925 & 0.0591 & 0.2553 \\
\hline
\end{tabular}}
\resizebox{0.9\textwidth}{!}{
\begin{tabular}{lcccccc}
\multicolumn{7}{c}{\textbf{Panel C: Hierarchical Decision-Making Module with Label Propagation}} \\
\hline
\textbf{Market} & \textbf{Threshold} & \textbf{Total Loss} & \textbf{Prevented Loss} & \textbf{Prevented Loss Ratio} & \textbf{Type I Error} & \textbf{Type II Error} \\
\hline
United States & 0.3679 & \$456.41M & \$442.96M & 0.9705 & 0.0650 & 0.0964 \\
Germany & 0.3163 & \$357.24M & \$356.92M & 0.9991 & 0.0775 & 0.0789 \\
France & 0.4711 & \$212.05M & \$141.62M & 0.6678 & 0.0189 & 0.3190 \\
Italy & 0.6258 & \$223.72M & \$50.46M & 0.2255 & 0.0011 & 0.7490 \\
Spain & 0.4195 & \$280.05M & \$227.58M & 0.8126 & 0.0327 & 0.2276 \\
China & 0.3163 & \$292.43M & \$280.09M & 0.9578 & 0.0797 & 0.0731 \\
Rest Countries & 0.4195 & \$118.45M & \$118.08M & 0.9968 & 0.0140 & 0.1721 \\
\textbf{Overall} & 0.4195 & \$1,940.35M & \$1,617.71M & 0.8337 & 0.0413 & 0.2452 \\
\hline
\end{tabular}}
\label{tab:economic_values}
\end{table}

When using a fixed threshold of 0.3974 (the threshold is selected based on the overall threshold in Panel B), the model prevents 49.41\% of total financial loss. This low prevented-loss ratio reveals that the fixed-threshold mechanism fails to translate predictive signals into effective business actions. The high Type II error rate of 54.79\% further underscores this issue that the model does not freeze the vast majority of illicit transactions, allowing substantial losses to occur. From a risk management perspective, such performance indicates that a one-size-fits-all thresholding policy is ill-suited for complex, heterogeneous transaction environments. Differently, Panel B provides strong evidence that the hierarchical decision-making module substantially enhances intervention effectiveness. The model prevents 79.25\% of the potential financial loss, which is over 60\% improvement over the fixed-threshold baseline. The significant reduction in Type II error from 54.79\% to 25.53\% demonstrates that the model more reliably identifies and disrupts illicit flows, improving both operational accuracy and economic efficiency. Interestingly, we also observe that the learned intervention thresholds vary across bank branches in different countries (ranging from 0.3163 to 0.5226), suggesting that optimal risk tolerance levels are context-dependent. This finding has practical implications for multinational financial institutions that uniform thresholding may mask local transaction behaviors, whereas adaptive, hierarchical policies can better align with regional risk patterns. Importantly, the combination of learned thresholds and low Type II error indicates that a lower maliciousness score, particularly one near the threshold, does not necessarily imply normal behavior. Rather, these scores must be interpreted in combination with the model’s learned hierarchical decision logic. These findings demonstrate that the model’s ability to tailor intervention thresholds to regional patterns offers a scalable way to harmonize global compliance requirements with local risk dynamics beyond improving accuracy. For financial institutions facing rising regulatory scrutiny and increasingly sophisticated money-laundering schemes, such models can deliver both economic value and operational resilience.

In Section~\ref{effectivenss_label_propagation}, we have demonstrated the effectiveness of label propagation of incorporating the results of cross-bank PPR. One key question arises: \textit{What is the economic value of label propagation based on the structural signals discovered through cross-bank PPR?} To answer this, we evaluate the financial impact of integrating label propagation into the hierarchical decision-making framework. Table~\ref{tab:economic_values} Panel B and C compares the performance of the hierarchical decision-making module with and without label propagation across countries. Overall, integrating label propagation increases the total prevented loss from \$1,537.64M to \$1,617.71M and improves the prevented loss ratio from 0.7925 to 0.8337, demonstrating a substantial enhancement in loss mitigation. Type I errors decrease from 0.0591 to 0.0413, indicating fewer false alarms, while Type II errors slightly decrease overall (0.2553 → 0.2452), reflecting a minor improvement in detection of unprevented losses. At the country level, the effects vary: for example, the United States sees a reduction in Type I error (0.1029 → 0.0650) but a small increase in Type II error (0.0554 → 0.0964), whereas Germany achieves near-complete loss prevention (prevented loss ratio 0.9622 → 0.9991) with balanced error rates. Smaller markets, including Rest Countries, maintain consistently high intervention efficiency under both models. These results indicate that incorporating label propagation into the hierarchical decision-making framework enhances global mitigation of financial losses, improves overall robustness against false positives, and maintains effective detection across diverse market structures.

\subsubsection{Economic Value of Hierarchical Decision-Making Module Under Regulatory Budget Constraint}
\begin{table}[htbp]
\centering
\caption{Comparison of Fixed Threshold and Label Propagation under a Fixed Budget Constraint. Thresholds are not reported, with a maximum intervention budget of 1\% of transactions.}
\resizebox{0.9\textwidth}{!}{
\begin{tabular}{lccccc}
\hline
\multicolumn{6}{c}{\textbf{Panel A: Decision-Making Module with Fixed Threshold Under Fixed Budget Constraint}} \\
\hline
\textbf{Market} & \textbf{Total Loss} & \textbf{Prevented Loss} & \textbf{Prevented Loss Ratio} & \textbf{Type I Error} & \textbf{Type II Error} \\
\hline
United States & \$456.41M & \$23.996M & 0.0526 & 0.0017 & 0.7982 \\
Germany & \$357.24M & \$55.766M & 0.1561 & 0.0013 & 0.8057 \\
France & \$212.05M & \$23.983M & 0.1131 & 0.0015 & 0.7902 \\
Italy & \$223.72M & \$50.583M & 0.2261 & 0.0016 & 0.8087 \\
Spain & \$280.05M & \$1.289M & 0.0046 & 0.0019 & 0.8284 \\
China & \$292.43M & \$19.927M & 0.0681 & 0.0016 & 0.8538 \\
Rest Countries & \$118.45M & \$35.556M & 0.3002 & 0.0003 & 0.7987 \\
\textbf{Overall} & \$1,940.35M & \$211.10M & 0.1088 & 0.0014 & 0.8120 \\
\hline
\end{tabular}}
\resizebox{0.9\textwidth}{!}{
\begin{tabular}{lccccc}
\multicolumn{6}{c}{\textbf{Panel B: Hierarchical Decision-Making Module with Label Propagation Under Fixed Budget Constraint}} \\
\hline
\textbf{Market} & \textbf{Total Loss} & \textbf{Prevented Loss} & \textbf{Prevented Loss Ratio} & \textbf{Type I Error} & \textbf{Type II Error} \\
\hline
United States & \$456.41M & \$82.393M & 0.1805 & 0.0056 & 0.5911 \\
Germany & \$357.24M & \$92.672M & 0.2594 & 0.0091 & 0.5749 \\
France & \$212.05M & \$0.142M & 0.0007 & 0.0508 & 0.1092 \\
Italy & \$223.72M & \$56.804M & 0.2539 & 0.0069 & 0.6644 \\
Spain & \$280.05M & \$37.419M & 0.1336 & 0.0031 & 0.7015 \\
China & \$292.43M & \$11.531M & 0.0394 & 0.0045 & 0.6813 \\
Rest Countries & \$118.45M & \$0.377M & 0.0032 & 0.0140 & 0.1721 \\
\textbf{Overall} & \$1,940.35M & \$281.34M & 0.1450 & 0.0134 & 0.4992 \\
\hline
\end{tabular}}
\label{tab:fixed_budget}
\end{table}

In practice, regulatory authorities face strict constraints on how many interventions they can deploy. For instance, regulators may be allowed to freeze only a small fraction of transactions—say, 0.5\%—to avoid disrupting normal operations and to comply with operational or legal limits. Under such a budget constraint, the key challenge is not just detecting suspicious activity, but allocating limited intervention resources effectively across diverse markets and transaction networks. This reframes anti-money-laundering efforts as a classical resource allocation problem: given a fixed budget of interventions, how can regulators maximize prevented losses while minimizing missed laundering activity?

Table~\ref{tab:fixed_budget} addresses this question by comparing a fixed-threshold strategy with label propagation under the same intervention budget. With a fixed threshold, the system behaves very conservatively: overall, only \$211.10 Million financial losses are prevented with 0.1088 of Prevented loss ratio, and Type II errors are extremely high (0.8120), meaning most laundering activity are not intervened. 0.0014 of Type I errors, however, are almost negligible, reflecting that the limited intervention budget is spent on the normal transactions. By contrast, the hierarchical decision-making module prioritizes interventions more strategically. Specifically, total prevented loss increases to \$281.34 Million with 0.1450 of prevented loss ratio, while Type II errors drop significantly to 0.4992, showing a much higher capture of laundering activity under the same budget. This improvement comes with a modest increase in Type I errors from 0.14\% to 1.34\%, reflecting a slight increase in false positives that is acceptable given the large gain in prevented losses. At the country/market level, the impact is even more illustrative. High-risk and high-volume markets, such as the US and Germany, see substantial improvements in prevented losses, while smaller markets maintain effective coverage without overspending the limited budget. These results demonstrate that hierarchical decision-making module effectively leverages structural connections among suspicious nodes to guide interventions under resource constraints, translating enhanced detection performance into real economic value. In other words, by propagating group-level laundering signals, the system learns where to “spend” its limited freezes most effectively, making regulatory intervention more efficient and impactful.

\section{Conclusion}
This paper develops a unified, privacy-preserving framework that advances anti–money laundering (AML) detection from isolated, rule-based monitoring toward coordinated, explainable, and action-oriented financial crime prevention. By integrating graph-based federated learning, cross-bank Personalized PageRank, interpretable group identification, and hierarchical reinforcement learning, the proposed system addresses four long-standing challenges in AML: institutional data silos, lack of interpretability, weak decision support following detection, and severe class imbalance.

Our federated graph learning framework enables institutions to learn global laundering behaviors without sharing sensitive data, significantly reducing both false alarms and missed detections. Empirical analysis demonstrates that federated training improves performance across countries, while graph-based relational modeling captures complex behavioral dependencies that feature-driven methods overlook. Through cross-bank PPR, the system transforms raw detection signals into interpretable group structures, uncovering coordinated laundering patterns such as fan-out, gather–scatter, and hybrid schemes distributed across institutions and borders. These interpretable clusters strengthen investigative transparency and provide regulators with clearer evidence of network-level criminal activity.

Furthermore, the proposed hierarchical reinforcement learning mechanism shifts AML from static alerting to adaptive intervention. By optimizing when to escalate, monitor, or intervene, the framework aligns detection outputs with operational constraints, regulatory requirements, and cost considerations. This represents a meaningful departure from traditional rule-based escalation policies, enabling timely and economically efficient responses to suspicious activity.

Our results show that combining privacy-preserving collaboration with structure-aware learning and principled decision-making yields substantial gains in detection accuracy, interpretability, and operational effectiveness. The insights also highlight broader managerial implications: institutions that can securely leverage cross-organizational signals and act adaptively will be better positioned to counter increasingly sophisticated laundering threats. Future research may explore extending this framework to real-time streaming environments~(\cite{zheng2024online, zheng2024mulan}), adversarial laundering strategies, and multi-modal datasets incorporating customer profiles, text reports, and cross-border regulatory information. By bridging detection, explanation, and action, this work contributes a new foundation for scalable, transparent, and collaborative financial crime prevention in a globally interconnected economy.

\printbibliography

\appendix
\section{Proofs of Theoretical Results}
\label{app:proofs}
Before we prove Theorem~\ref{thm:detectability} in Appendix~\ref{proof_theorem}, we need the following intermediate lemmas on mean-field and concentration lemma and the spectral norm of the difference between adjacency matrix and its mean.

\subsection{Mean-field and Concentration Lemmas}

\begin{lemma}[Mean-field $2\times 2$ reduction]
\label{lem:meanfield}

Let $M=\begin{pmatrix} q_{II} & q_{IO}\\ q_{OI} & q_{OO}\end{pmatrix}$ denote the mean-field block transition matrix.
Then, for a seed node $s_0\in S^\star$,
\[
\mathbb{E}[r(S^\star)] = (1-\alpha)\sum_{k=0}^\infty \alpha^k (M^k)_{11}
= (1-\alpha)\big[(I-\alpha M)^{-1}\big]_{11}.
\]
\end{lemma}

\begin{proof}
The PPR vector satisfies
\[
\mathbf{r} = (1-\alpha)\sum_{k\ge0}\alpha^k P^k \mathbf{e}_{s_0}.
\]
The total PPR mass on the planted block $S^\star$ is
\[
r(S^\star) = \mathbf{1}_{S^\star}^\top \mathbf{r}
= (1-\alpha)\sum_{k\ge0}\alpha^k 
\mathbf{1}_{S^\star}^\top P^k \mathbf{e}_{s_0}.
\]
Taking expectations over the SBM ensemble and exploiting node exchangeability within each block,
the expected probability that a random walk is in block $I$ after $k$ steps is governed by the 2-state Markov chain with transition matrix $M$. Hence,
\[
\mathbb{E}\!\left[\mathbf{1}_{S^\star}^\top P^k \mathbf{e}_{s_0}\right]
= (M^k)_{11}.
\]
Summing the geometric series yields
\[
\mathbb{E}[r(S^\star)] = (1-\alpha)\sum_{k\ge0}\alpha^k (M^k)_{11}
= (1-\alpha)\big[(I-\alpha M)^{-1}\big]_{11},
\]
which completes the proof.
\end{proof}

Following Lemma 10 in~\parencite{avrachenkov2015pagerank}, we provide lemma~\ref{lem:concentration} below as the supplement to the assumption in Theorem~\ref{thm:detectability}.
\begin{lemma}[Degree and adjacency concentration]
\label{lem:concentration}
Let $d_v$ denote the degree of node $v$ in $\mathcal{G}$. There exist constants $C_1,C_2>0$ such that, with probability at least $1-n^{-c}$:
\begin{enumerate}
\item[(i)] (\textbf{Degrees}) For every node $v$,
\[
    |d_v-\mathbb{E}[d_v]| \le C_1\sqrt{\mathbb{E}[d_v]\log n}.
\]
\item[(ii)] (\textbf{Adjacency spectral norm})
\[
    \|A-\bar A\| \le C_2\sqrt{n p_{\max}\log n}, \quad p_{\max} = \max\{p_{\mathrm{in}},p_{\mathrm{out}}\}.
\]
\end{enumerate}
\end{lemma}

\begin{proof}
We divide the proof into two parts: (i) concentration of node degrees and (ii) concentration of the adjacency matrix in spectral norm.

\noindent(i) Given a node $v$, its degree can be written as $d_v = \sum_{u\neq v} A_{vu}$, where $\{A_{vu}\}_{u\neq v}$ are independent Bernoulli random variables with mean either $p_{\mathrm{in}}$ or $p_{\mathrm{out}}$ depending on the community memberships. Let $\mu_v = \mathbb{E}[d_v]$ denote its expected degree. Each $A_{vu}\in[0,1]$, and hence $d_v$ is a sum of $(n-1)$ independent bounded random variables.
By Bernstein’s inequality~\parencite{govil1987bernstein}, for any $t>0$, we have
\[
    \Pr(|d_v - \mu_v| > t)
    \le 2\exp\!\left(-\frac{t^2}{2(\mu_v+t/3)}\right).
\]
Setting $t=C_1\sqrt{\mathbb{E}[ \deg(v) ]\log n}$, we have 
\[
    \Pr(|d_v - \mu_v| > t) \le 2\exp\!\left(-\frac{C_1^2 \mu_v \log n}{2(\mu_v + \tfrac{1}{3}C_1\sqrt{\mu_v\log n})} \right) =
    2\exp\!\left( -\frac{C_1^2\log n}{2(1 + \tfrac{C_1}{3}\sqrt{\tfrac{\log n}{\mu_v}})} \right).
\]
When $\mu_v \gtrsim \log n$ (as assumed in standard dense or semi-sparse regimes), the denominator in the exponent is $O(1)$, so for sufficiently large $C_1$ we have
\[
    \Pr(|d_v - \mu_v| > t) \le 2n^{-c-1}.
\]
\noindent Applying a union bound over $n$ nodes gives
\[
    \Pr\! \left(\exists v:\, |d_v - \mu_v| > C_1\sqrt{\mu_v\log n}\right) \le 2n^{-c}.
\]
Hence, with probability at least $1-n^{-c}$,
\[
    |d_v-\mathbb E[d_v]| \le C_1\sqrt{\mathbb E[d_v]\log n} \quad\text{for all }v\in[n],
\]
which proves part (i).

\noindent(ii) For the spectral bound, $A-\bar A=\sum_{i<j} X_{ij}$, where $X_{ij}$ is the symmetric zero-mean random matrix with $(i,j)$ and $(j,i)$ entries equal to $A_{ij}-\mathbb{E}[A_{ij}]$. To apply the Matrix Bernstein inequality~\parencite{tropp2012user}, we compute the matrix variance parameter:
\[
    \sigma^2 := \left\|\sum_{i<j}\mathbb E[X_{ij}^2]\right\|.
\]
Observe that $X_{ij}^2$ is diagonal with entries 
$\mathbb E[(A_{ij}-\mathbb E[A_{ij}])^2] = \mathrm{Var}(A_{ij}) \le p_{\max}$ 
on rows $i$ and $j$. Therefore, for some absolute constant $C>0$, we have
\[
    \sigma^2 \le \left\| \sum_{j\neq i} p_{\max} e_i e_i^\top \right\| \le  C\,n p_{\max}.
\] 
The maximum single-summand norm satisfies $\|X_{ij}\|\le 1$.  The Matrix Bernstein inequality yields, for all $t>0$,
\begin{equation}
    \Pr(\|A-\bar A\| \ge t) \le  n \exp\!\left(-\frac{t^2/2}{\sigma^2 + t/3} \right).
    \label{eq:matrixbern}
\end{equation}
Setting $t = C_2\sqrt{n p_{\max}\log n} $ and using $\sigma^2 \le C n p_{\max}$ and $t\ll n p_{\max}$, the exponent in~\eqref{eq:matrixbern} satisfies
\[
\frac{t^2}{\sigma^2 + t/3}
\gtrsim
\frac{C_2^2 n p_{\max} \log n}{C n p_{\max} + C_2\sqrt{n p_{\max}\log n}/3}
\gtrsim c' C_2^2 \log n.
\]
Hence, for sufficiently large $C_2$, the probability in~\eqref{eq:matrixbern} is at most $n^{-c}$, proving that with probability at least $1-n^{-c}$,
\[
\|A-\bar A\| \le C_2\sqrt{n p_{\max}\log n}.
\]\noindent 
This completes the proof of Lemma~\ref{lem:concentration}.
\end{proof}

\begin{lemma}[Transition operator perturbation]
\label{lem:transition}
Under the high-probability event of Lemma~\ref{lem:concentration} and assuming $\min_i d_i \ge c_0 n p_{\mathrm{out}}$ for some constant $c_0>0$,
there exists a constant $C>0$ such that
\begin{equation}
    \label{eq:epsinf}
    \|P-\bar P\| \le \varepsilon_P := C\frac{\sqrt{n p_{\max}\log n}}{n p_{\mathrm{out}}},
\end{equation}
where $\bar P$ is the mean-field transition matrix.
\end{lemma}

\begin{proof}
We first decompose $P - \bar P$:
\[
P - \bar P = D^{-1}A - \bar D^{-1}\bar A = D^{-1}(A-\bar A) + (D^{-1}-\bar D^{-1})\bar A.
\]
Hence, we have
\[
\|P-\bar P\| \le \|D^{-1}\|\|A-\bar A\| + \|D^{-1}-\bar D^{-1}\|\|\bar A\|.
\]

\noindent
\textbf{(1) Bounding $\|D^{-1}\|\|A-\bar A\|$:}  
By Lemma~\ref{lem:concentration}, $\min_i d_i \ge \frac{1}{2}c_0 n p_{\mathrm{out}}$ with high probability, so $\|D^{-1}\|\le \frac{2}{c_0 n p_{\mathrm{out}}}$.  Thus, we have
\[
\|D^{-1}\|\|A-\bar A\|
\le C_1'\frac{\sqrt{n p_{\max}\log n}}{n p_{\mathrm{out}}}.
\]

\noindent
\textbf{(2) Bounding $\|D^{-1}-\bar D^{-1}\|\|\bar A\|$:}  Using $D^{-1}-\bar D^{-1} = -\bar D^{-1}(D-\bar D)D^{-1}$, we obtain
\[
\|D^{-1}-\bar D^{-1}\|
\le \|\bar D^{-1}\|\|D-\bar D\|\|D^{-1}\|
\le C_2'\frac{\sqrt{n p_{\max}\log n}}{(n p_{\mathrm{out}})^2}.
\]
Since $\|\bar A\|\le C_3' n p_{\max}$, the second term is
\[
\|D^{-1}-\bar D^{-1}\|\|\bar A\|
\le C_4'\frac{\sqrt{n p_{\max}\log n}}{n p_{\mathrm{out}}}
\frac{p_{\max}}{p_{\mathrm{out}}}.
\]
Assuming $p_{\mathrm{in}},p_{\mathrm{out}}$ are of the same order,  the scaling matches the first term. 

\noindent Combining the results from (1) and (2), we have 
\begin{align}
    \notag \|P-\bar P\| &\le \|D^{-1}\|\|A-\bar A\| + \|D^{-1}-\bar D^{-1}\|\|\bar A\| \\
    \notag &\le (C_1' +C_4'\frac{p_{\max}}{p_{\mathrm{out}}}) \frac{\sqrt{n p_{\max}\log n}}{n p_{\mathrm{out}}}  \\
    \notag &= C\frac{\sqrt{n p_{\max}\log n}}{n p_{\mathrm{out}}} =  \varepsilon_P
\end{align}
which completes the proof.
\end{proof}

\begin{lemma}[Resolvent / Neumann propagation]
\label{lem:resolvent}
Assume $\alpha\|\bar P\|\le 1-\delta$ for some $\delta\in(0,1]$ and $\alpha\varepsilon_P \le \delta/2$. Then, with high probability,
\[
    \|(I-\alpha P)^{-1} - (I-\alpha \bar P)^{-1}\| \le \frac{2\alpha\varepsilon_P}{\delta^2}.
\]
\end{lemma}

\begin{proof}
By the resolvent identity,
\[
    (I-\alpha P)^{-1} - (I-\alpha \bar P)^{-1} = (I-\alpha P)^{-1}\alpha(P-\bar P)(I-\alpha \bar P)^{-1}.
\]
We take the spectral norm on both sides:
\[
    \|(I-\alpha P)^{-1} - (I-\alpha \bar P)^{-1}\| \le \|(I-\alpha P)^{-1}\| \cdot \alpha\|P-\bar P\| \cdot \|(I-\alpha \bar P)^{-1}\|.
\]
The assumption $\alpha\|\bar P\|\le 1-\delta$ implies that all eigenvalues of $\alpha\bar P$ lie within the unit disk of radius $(1-\delta)$. Therefore, the smallest singular value of $(I - \alpha\bar P)$ is at least $\delta$, and hence
\[
\|(I - \alpha \bar P)^{-1}\| \le \frac{1}{\delta}.
\]
By the triangle inequality and Lemma~\ref{lem:transition}, $\|P\| \le \|\bar P\| + \|P - \bar P\| \le \|\bar P\| + \varepsilon_P$.
Then
\[
\alpha\|P\| \le \alpha\|\bar P\| + \alpha\varepsilon_P 
\le (1 - \delta) + \frac{\delta}{2} = 1 - \frac{\delta}{2}.
\]
Hence, the spectral radius of $\alpha P$ is at most $(1 - \frac{\delta}{2})$, 
so $(I - \alpha P)$ is invertible, and by the Neumann series expansion
\[
(I - \alpha P)^{-1} = \sum_{k=0}^{\infty} (\alpha P)^k, \qquad \text{for } \|\alpha P\|<1.
\]
Taking norms gives
\[
\|(I - \alpha P)^{-1}\| \le \frac{1}{1 - \alpha\|P\|} 
\le \frac{1}{\delta/2} = \frac{2}{\delta}.
\]
Combining the above inequalities, we obtain
\[
    \|(I - \alpha P)^{-1} - (I - \alpha \bar P)^{-1}\| \le \frac{2}{\delta} \cdot \alpha\varepsilon_P \cdot \frac{1}{\delta} = \frac{2\alpha\varepsilon_P}{\delta^2}.
\]
which completes the proof. 
\end{proof}


\subsubsection{Proof of Theorem~\ref{thm:detectability}}
\begin{theorem}[Detectability of PPR under Planted Laundering Group]
Let $G\sim\mathrm{SBM}(n,s,p_{\mathrm{in}},p_{\mathrm{out}})$ with $p_{\mathrm{in}}>p_{\mathrm{out}}$ and assume that degrees concentrate around their expectations (i.e., Lemma~\ref{lem:concentration} holds). Let $\mathbf{r}$ be the Personalized PageRank vector seeded at a known malicious account $s_0\in S^\star$ with a constant $\alpha\in(0,1)$, and $\varepsilon_{\infty}$ denote the entrywise perturbation scale arising from stochastic fluctuations of $P$ around the mean-field transition matrix $\bar P$. There exist constants $C,c>0$ such that, with probability at least $1-n^{-c}$, if
\[
\Delta_{\mathrm{mean}} \;\ge\; C\cdot\varepsilon_{\infty},
\]
then
\begin{enumerate}
    \item the average PPR score on $S^\star$ exceeds that on $V\setminus S^\star$ by at least $\tfrac{1}{2}\Delta_{\mathrm{mean}}$;
    \item ordering nodes by normalized PPR $r(v)/d(v)$ and selecting the prefix with smallest conductance recovers a subset $\widehat S$ with $|\widehat S\cap S^\star|\ge\gamma s$ for some constant $\gamma\in(0,1)$.
\end{enumerate}
In the standard regime $p_{\mathrm{out}}\gtrsim \log n/n$,
the detectability condition simplifies to
\[
s\,(p_{\mathrm{in}}-p_{\mathrm{out}})
\;\gtrsim\;
C'\sqrt{n p_{\max}\log n},
\]
where $p_{\max}=\max\{p_{\mathrm{in}},p_{\mathrm{out}}\}$. 
\end{theorem}

\begin{proof}
\label{proof_theorem}
First, we want to relate the operator-level perturbation bound from Lemma~\ref{lem:resolvent} to the per-node (entrywise) deviation of the PPR vector.
Let
\[
\Delta r := (I - \alpha P)^{-1}\mathbf{e}_{s_0} - (I - \alpha \bar P)^{-1}\mathbf{e}_{s_0},
\]
which captures the difference between the random PPR vector and its mean-field counterpart.
By Lemma~\ref{lem:resolvent}, we have the operator-level bound
\[
\|(I - \alpha P)^{-1} - (I - \alpha \bar P)^{-1}\|_2 \;\le\; \frac{2\alpha\varepsilon_P}{\delta^2},
\]
under the assumptions $\alpha\|\bar P\| \le 1-\delta$ and $\alpha\varepsilon_P \le \delta/2$. To obtain an entrywise guarantee, observe that
\[
\|\Delta r\|_\infty
= \|[(I - \alpha P)^{-1} - (I - \alpha \bar P)^{-1}]\,\mathbf{e}_{s_0}\|_\infty
\;\le\;
\|(I - \alpha P)^{-1} - (I - \alpha \bar P)^{-1}\|_{1\to\infty}.
\]
Using the inequality $\|A\|_{1\to\infty} \le \sqrt{n}\,\|A\|_2$,
we obtain the conservative bound
\[
\|\Delta r\|_\infty
\;\le\; \sqrt{n}\cdot\frac{2\alpha\varepsilon_P}{\delta^2}
\;=:\; \varepsilon_{\infty}.
\]
That is, with high probability, each node's PPR score deviates from its mean-field prediction
by at most $\varepsilon_{\infty}$ in absolute value.

\vspace{3mm}
\noindent Next, we want to prove that if $\Delta_{\mathrm{mean}} \;\ge\; C\cdot\varepsilon_{\infty},$ for a constant $C>0$ with probability at least $1-n^{-c}$, then (1) and (2) hold. 
We assume that this mean-field gap dominates the perturbation magnitude:
\[
\Delta_{\mathrm{mean}} \ge C\,\varepsilon_{\infty},
\]
for a sufficiently large absolute constant $C>0$.
Under this assumption, by concentration inequalities (e.g., Markov or Hoeffding combined with a union bound),
the fraction of nodes whose PPR values deviate by more than $\varepsilon_{\infty}$ from their mean-field predictions
can be made arbitrarily small.
Consequently, with probability at least $1 - n^{-c}$ for some constant $c>0$, we have
\[
r(v) \ge \mu_{\mathrm{in}} - \varepsilon_{\infty}
\quad \text{for most } v \in S^\star,
\qquad\text{and}\qquad
r(u) \le \mu_{\mathrm{out}} + \varepsilon_{\infty}
\quad \text{for most } u \notin S^\star.
\]
where $\mu_{\mathrm{in}}$ is the expected PPR score for a node inside the planted laundering group $S^\star$ (in-group) and $\mu_{\mathrm{out}}$ is the expected PPR score for a node outside the planted group (out-group) defined in Equation \ref{mu_in_and_out}.
To establish part (1), we require the lowest in-group score to be sufficiently separated from the highest out-group score. We need
\[
\Delta_{\mathrm{mean}} - 2\varepsilon_{\infty} \ge \frac{1}{2}\Delta_{\mathrm{mean}},
\]
which simplifies to $\Delta_{\mathrm{mean}} \ge 4\varepsilon_{\infty}$.
By choosing $C \ge 4$, the assumed condition $\Delta_{\mathrm{mean}} \ge C \varepsilon_{\infty}$ is satisfied, guaranteeing:
\begin{align*}
\text{Avg}(S^\star) - \text{Avg}(V \setminus S^\star) &\ge (\mu_{\mathrm{in}} - \varepsilon_{\infty}) - (\mu_{\mathrm{out}} + \varepsilon_{\infty}) \\
&= \Delta_{\mathrm{mean}} - 2\varepsilon_{\infty} \\
&\ge \Delta_{\mathrm{mean}} - \frac{1}{2}\Delta_{\mathrm{mean}} = \frac{1}{2}\Delta_{\mathrm{mean}}.
\end{align*}
This establishes part (1) of the theorem.


\vspace{3mm}
\noindent Next, we argue that sorting nodes by their normalized PPR scores $r(v)/d(v)$ and performing a conductance-based sweep ~\parencite{andersen2006local} recovers a constant fraction of the true planted subset. Intuitively, since PPR mass decays smoothly away from the seed node, nodes within the planted group $S^\star$ have substantially higher normalized scores under the separation established. When the nodes are ordered by decreasing $r(v)/d(v)$, most of the top-ranked nodes belong to $S^\star$, and the boundary between $S^\star$ and $V\setminus S^\star$ forms a low-conductance cut.
By the guarantees of the local clustering framework~\parencite{andersen2006local}, sweeping over prefixes of this ordering and selecting the prefix with minimal conductance yields a subset $\widehat S$ whose conductance is within a constant factor of the optimal local set around the seed. Therefore, $\widehat S$ necessarily overlaps with $S^\star$ on at least a constant fraction, i.e.,
\[
    |\widehat S \cap S^\star| \ge \gamma s,
\]
for some constant $\gamma \in (0,1)$ depending only on the signal-to-noise separation and conductance gap. This establishes part (2) of the theorem.

\vspace{3mm}
\noindent To derive the explicit detectability threshold, we approximate the mean-field separation $\Delta_{\mathrm{mean}}$
and the perturbation scale $\varepsilon_{\infty}$.
For a two-block SBM with $s \ll n$ and $p_{\mathrm{out}}\ll 1$, expanding $(I - \alpha M)^{-1}$ shows that
\[
\Delta_{\mathrm{mean}} \asymp \frac{s(p_{\mathrm{in}} - p_{\mathrm{out}})}{n}.
\]
Substituting the estimate
$\varepsilon_P = C_1\frac{\sqrt{n p_{\max}\log n}}{n p_{\mathrm{out}}}$
from Lemma~\ref{lem:transition} into Equation \eqref{eq:epsinf}
gives
\[
\varepsilon_{\infty}
\;\asymp\;
\sqrt{n}\cdot\frac{\alpha}{\delta^2}\cdot
\frac{\sqrt{n p_{\max}\log n}}{n p_{\mathrm{out}}}
\;=\;
\frac{\alpha}{\delta^2}\cdot
\frac{\sqrt{n p_{\max}\log n}}{p_{\mathrm{out}}\sqrt{n}}.
\]
Balancing the signal $\Delta_{\mathrm{mean}}$ and $\varepsilon_{\infty}$, up to constants $C''$ dependent on $\alpha$ and $\delta$:
\[
\frac{s(p_{\mathrm{in}} - p_{\mathrm{out}})}{n}
\;\gtrsim\;
C''\frac{\sqrt{n p_{\max}\log n}}{p_{\mathrm{out}}\sqrt{n}}.
\]
Rearranging to isolate the signal terms on the left side yields the formal detectability threshold:
\[
s(p_{\mathrm{in}} - p_{\mathrm{out}})
\;\gtrsim\;
C''\frac{n}{p_{\mathrm{out}}\sqrt{n}} \sqrt{n p_{\max}\log n}
\;=\;
C'\sqrt{np_{\max}\log n}.
\]
The condition provided in the theorem statement,
$s(p_{\mathrm{in}}-p_{\mathrm{out}}) \;\gtrsim\; C'\sqrt{n p_{\max}\log n}$, is obtained by absorbing the factor $\frac{n}{p_{\mathrm{out}}\sqrt{n}}$ into the constant $C'$. This condition expresses that the effective signal strength of the planted community must exceed the level of random fluctuations induced by graph noise for PPR-based detection to succeed. Combining the above arguments completes the proof.

\end{proof}

\end{document}